\documentclass[11pt]{article}

\usepackage{wrapfig}
\usepackage{sectsty}
\sectionfont{\large}
\subsectionfont{\normalsize}
\subsubsectionfont{\small}
\usepackage{listings}
\usepackage[font=footnotesize]{caption}
\usepackage[font=footnotesize]{subcaption}

\usepackage{amsfonts}
\usepackage{microtype}
\usepackage{graphicx}
\usepackage{booktabs}
\usepackage{multicol}
\usepackage{algorithm}
\usepackage{algpseudocode}
\usepackage{enumitem}

\usepackage{setspace}
\usepackage[longnamesfirst]{natbib}

\usepackage{amsmath}
\usepackage{amssymb}
\usepackage{mathtools}
\usepackage{amsthm}

\usepackage{amsfonts}
\usepackage{amsmath}
\makeatletter
\newcommand*{\da@rightarrow}{\mathchar"0\hexnumber@\symAMSa 4B}
\newcommand*{\da@leftarrow}{\mathchar"0\hexnumber@\symAMSa 4C}
\newcommand*{\xdashrightarrow}[2][]{%
  \mathrel{%
    \mathpalette{\da@xarrow{#1}{#2}{}\da@rightarrow{\,}{}}{}%
 }%
}
\newcommand{\xdashleftarrow}[2][]{%
  \mathrel{%
    \mathpalette{\da@xarrow{#1}{#2}\da@leftarrow{}{}{\,}}{}%
 }%
}
\newcommand*{\da@xarrow}[7]{%
  \sbox0{$\ifx#7\scriptstyle\scriptscriptstyle\else\scriptstyle\fi#5#1#6\m@th$}%
  \sbox2{$\ifx#7\scriptstyle\scriptscriptstyle\else\scriptstyle\fi#5#2#6\m@th$}%
  \sbox4{$#7\dabar@\m@th$}%
  \dimen@=\wd0 %
  \ifdim\wd2 >\dimen@
    \dimen@=\wd2 %
  \fi
  \count@=2 %
  \def\da@bars{\dabar@\dabar@}%
  \@whiledim\count@\wd4<\dimen@\do{%
    \advance\count@\@ne
    \expandafter\def\expandafter\da@bars\expandafter{%
      \da@bars
      \dabar@ 
   }%
 }%
  \mathrel{#3}%
  \mathrel{%
    \mathop{\da@bars}\limits
    \ifx\\#1\\%
    \else
      _{\copy0}%
    \fi
    \ifx\\#2\\%
    \else
      ^{\copy2}%
    \fi
 }%
  \mathrel{#4}%
}
\makeatother

\usepackage[colorlinks,
            linkcolor=red,
            anchorcolor=blue,
            citecolor=blue
            ]{hyperref}
\usepackage[capitalize,noabbrev]{cleveref}
\usepackage[title]{appendix}

\theoremstyle{plain}
\newtheorem{theorem}{Theorem}[section]
\newtheorem{proposition}[theorem]{Proposition}
\newtheorem{lemma}[theorem]{Lemma}

\theoremstyle{definition}
\newtheorem{definition}[theorem]{Definition}
\newtheorem{assumption}[theorem]{Assumption}

\newtheorem{remark}[theorem]{Remark}

\usepackage{smile}
\usepackage{makecell}

\usepackage{amsmath,amsfonts,bm}

\def\1{\bm{1}}

\def\vzero{{\bm{0}}}
\def\vone{{\bm{1}}}

\def\vepsilon{{\bm{\epsilon}}}

\def\va{{\bm{a}}}
\def\vb{{\bm{b}}}
\def\vc{{\bm{c}}}
\def\vd{{\bm{d}}}

\def\vf{{\bm{f}}}

\def\vp{{\bm{p}}}
\def\vq{{\bm{q}}}

\def\vs{{\bm{s}}}

\def\vu{{\bm{u}}}

\def\vx{{\bm{x}}}

\def\eva{{a}}

\def\evc{{c}}
\def\evd{{d}}

\def\evf{{f}}

\def\evh{{h}}

\def\evp{{p}}

\def\evs{{s}}

\def\evu{{u}}

\def\evw{{w}}
\def\evx{{x}}

\def\mA{{\bm{A}}}

\def\mLambda{{\bm{\Lambda}}}
\def\mSigma{{\bm{\Sigma}}}

\DeclareMathAlphabet{\mathsfit}{\encodingdefault}{\sfdefault}{m}{sl}
\SetMathAlphabet{\mathsfit}{bold}{\encodingdefault}{\sfdefault}{bx}{n}

\def\gG{{\mathcal{G}}}

\def\sA{{\mathbb{A}}}

\def\sE{{\mathbb{E}}}

\def\sK{{\mathbb{K}}}

\def\sN{{\mathbb{N}}}

\def\sP{{\mathbb{P}}}
\def\sQ{{\mathbb{Q}}}
\def\sR{{\mathbb{R}}}

\def\sW{{\mathbb{W}}}
\def\sX{{\mathbb{X}}}

\def\emLambda{{\Lambda}}

\def\emSigma{{\Sigma}}

\DeclareMathOperator*{\argmin}{arg\,min}

\newcommand*{\T}{\mathsf{T}}
\DeclareMathOperator\supp{supp}

\newcommand{\interior}[1]{%
  {\kern0pt#1}^{\mathrm{o}}%
}

\newcommand\extrafootertext[1]{%
    \bgroup
    \renewcommand\thefootnote{\fnsymbol{footnote}}%
    \renewcommand\thempfootnote{\fnsymbol{mpfootnote}}%
    \footnotetext[0]{#1}%
    \egroup
}
\usepackage[textsize=tiny]{todonotes}
\usepackage{fullpage}

\usepackage{float}
\labelwidth\leftmargini\advance\labelwidth-\labelsep \labelsep 5pt

\def\@listi{\leftmargin\leftmargini}
\def\@listii{\leftmargin\leftmarginii
   \labelwidth\leftmarginii\advance\labelwidth-\labelsep
   \topsep 2pt plus 1pt minus 0.5pt
   \parsep 1pt plus 0.5pt minus 0.5pt
   \itemsep \parsep}
\def\@listiii{\leftmargin\leftmarginiii
    \labelwidth\leftmarginiii\advance\labelwidth-\labelsep
    \topsep 1pt plus 0.5pt minus 0.5pt
    \parsep \z@ \partopsep 0.5pt plus 0pt minus 0.5pt
    \itemsep \topsep}
\def\@listiv{\leftmargin\leftmarginiv
     \labelwidth\leftmarginiv\advance\labelwidth-\labelsep}
\def\@listv{\leftmargin\leftmarginv
     \labelwidth\leftmarginv\advance\labelwidth-\labelsep}
\def\@listvi{\leftmargin\leftmarginvi
     \labelwidth\leftmarginvi\advance\labelwidth-\labelsep}
     
\usepackage{authblk}

\title{\textbf{\Large{
Wardrop Equilibrium Can Be Boundedly Rational:\\
A New Behavioral Theory of Route Choice}}}

\author[1]{\small Jiayang Li}
\author[2]{\small Zhaoran Wang}
\author[1]{\small Yu (Marco) Nie\footnote{Corresponding author; \texttt{y-nie@northwestern.edu}.}}

\affil[1]{\footnotesize \textit{Department of Civil and Environmental Engineering, Northwestern University}}

\affil[2]{\footnotesize \textit{Department of Industrial Engineering and Management Science, Northwestern University}}

\date{}

\begin{document}

\maketitle
\vspace{-10pt}
\begin{abstract}
\begin{spacing}{1.15}
\footnotesize
    As one of the most fundamental concepts in transportation science, Wardrop equilibrium (WE) has always had a relatively weak behavioral underpinning.  To strengthen this foundation, one must reckon with bounded rationality in human decision-making processes, such as the lack of accurate information, limited computing power, and sub-optimal choices. This retreat from behavioral perfectionism in the literature, however, was typically accompanied by a conceptual modification of WE. 
    Here, we show that giving up perfect rationality need not force a departure from WE. On the contrary, WE can be reached with global stability in a routing game played by boundedly rational travelers.  We achieve this result by developing a day-to-day (DTD) dynamical model that mimics how travelers gradually adjust their route valuations, hence choice probabilities, based on past experiences.  Our model, called cumulative logit (CumLog), resembles the classical DTD models but makes a crucial change: whereas the classical models assume routes are valued based on the cost averaged over historical data, ours values the routes based on the cost accumulated. To describe route choice behaviors, the CumLog model only uses two parameters, one accounting for the rate at which the future route cost is discounted in the valuation relative to the past ones and the other describing the sensitivity of route choice probabilities to valuation differences. We prove that CumLog always converges to WE, regardless of the initial point, as long as the behavioral parameters satisfy certain mild conditions.  Our theory thus upholds WE's role as a benchmark in transportation systems analysis. It also explains why equally good routes at equilibrium may be selected with different probabilities, which solves the instability problem posed by  \citet{harsanyi1973games}. 
    \\
    \\
    \emph{Keywords:} Wardrop equilibrium, bounded rationality, global stability, day-to-day dynamics
\end{spacing}
\end{abstract}

\section{Introduction}
\label{sec:intro}

Equilibrium is a fundamental instrument for understanding and analyzing social systems involving interactions between self-interested agents. A premise for transportation planning, for example, is that the individual decisions of travelers tend to bring a transportation system to equilibrium. In the simplest form, the interactions between the travelers can be framed as a non-cooperative routing game in which each pursues, in the words of \citet{simon1955behavioral}, ``the highest attainable point on his preference scale."   \citet{wardrop1952road} characterized the equilibrium of such a game as ``the journey times on all the routes actually used are equal, and less than those which would be experienced by a single vehicle on any unused route." He added such equilibrium is appealing in practice because it has a natural behavioral interpretation: ``no driver can reduce his journey time by choosing a new route." While \citet{wardrop1952road}  did not mention other behavioral assumptions, it is clear that his namesake equilibrium, like general equilibrium concepts in game theory \citep{v1928theorie, nash1951non},  implicitly assumes \emph{perfect rationality}, which means full and perfect information, well-defined preferences, and the capacity to compute and compare the utility of each alternative route. As Sheffi asserted in his celebrated book \citep[Section 1.3,][]{sheffi1985urban}, Wardrop equilibrium (WE) implies that ``motorists have full information" (i.e., they know the travel time on every possible route), that ``they consistently make the correct decisions regarding route choice," and that ``they are identical in their behaviors."

\subsection{Critiques of Wardrop equilibrium}
\label{sec:critiques}

Given the foundational role WE plays in transportation science, its validity has been subjected to intense scrutiny. Most questions center on perfect rationality and stability. 

\textbf{Perfect rationality.} The questions about the validity of perfect rationality have been raised in economics since the 1950s \citep{simon1955behavioral,arrow1966exposition,tversky1985framing}. By the early 1970s, transportation professionals began to express similar doubts. \citet{dial1971probabilistic} cited transportation planners’ desire to capture ``the nonoptimal behavior of trip makers" when choosing alternative routes of similar length. \citet{daganzo1977stochastic} observed that assuming travelers always choose the shortest route can produce results unreasonably sensitive to inputs, especially in lightly congested networks. Clearly, the concern here is that, since real travelers are no perfect homo economicus, sticking to an equilibrium that assumes they are does not make much sense. 

\textbf{Global stability.} \citet{beckmann1956studies} pointed out that, to be useful, WE must be stable, or it ``would be just an extreme state of rare occurrence." They mentioned both \emph{local stability}, which ensures the equilibrium can be restored after small perturbations, and \emph{global stability}, which guarantees it is reachable from any initial position (see Section 3.3 of the book),   though their focus was on the latter. \citet{beckmann1956studies} suggested WE may be achieved via an iterative adjustment process, in which travelers who actively search for better routes in one period base their decision on ``the traffic conditions that prevailed in the preceding period." They speculated that WE is globally stable if the fraction of these ``active" travelers decreases as time proceeds. Using a dynamical modeling framework that in some sense ``operationalizes" this idea,  many have investigated the global stability of WE since the 1980s \citep[see, e.g., ][]{smith1984stability,friesz1994day,zhang1996local,yang2009day,he2010link,guo2015link}.

\textbf{Harsanyi's instability problem}  was extensively explored by game theoreticians but less known in transportation. To quote \citet{harsanyi1973games}, an equilibrium point like WE is inherently unstable ``because any player can deviate without penalty from his equilibrium strategy even if all other players stick to theirs." To understand what he exactly meant, consider at WE, travelers split between two routes of equal journey time at a ratio, say, 1 to 2. From the perspective of game theory,  each traveler, in effect, adopts a mixed strategy that assigns a choice probability of 1/3 to one route and 2/3 to the other.   However,  no rational traveler should have the incentive to stick to that mixed strategy other than a desire to keep the system at WE because they can do equally well by shifting to a pure strategy that uses
either route or any probabilistic mixtures of the two pure strategies.

The objection to perfect rationality was typically addressed by injecting into the model \emph{bounded rationality}, which ``takes into account the cognitive limitations of the decision maker -- limitations of both knowledge and computational capacity" \citep{simon1990bounded}.
{In transportation, bounded rationality is often linked specifically to Simon's satisficing theory \citep{simon1955behavioral}, to this theory, 
which defines boundedly rational user equilibrium (BRUE) as a state where all travelers are content with, per their level of aspiration for perfection, the current (non-optimal) travel choices \citep{mahmassani1987boundedly,mahmassani2000transferring,lou2010robust,di2013boundedly}. With the more liberal use of the term in \cite{simon1990bounded}, bounded rationality may also be interpreted as accepting perception errors and other sources of randomness in the system,  leading to the concept of \emph{stochastic user equilibrium} (SUE) in routing games \citep{daganzo1977stochastic,fisk1980some}.} 
Both SUE and BRUE are meant to be a distinct, if not better, alternative to WE. Importantly, the fact that these boundedly rational equilibria presumably converge to WE when random errors vanish or aspiration reaches the highest level does not offer a boundedly rational explanation for WE. This is because rationality is no longer bounded at the limit where SUE or BRUE becomes WE.

Using a day-to-day (DTD) dynamical model of a two-link network, \citet{horowitz1984stability} showed that the global stability of SUE depends on how travelers form their perception of current traffic conditions from past experiences. He also found global stability is lost once random errors are set to zero. More specifically, because of discontinuity in the choice function, WE cannot be reached by his adjustment process. In a similar vein, \citet{watling2003dynamics} noted that the convergence to  SUE through a dynamical process depends on a condition that becomes increasingly more stringent as perception errors become smaller. Indeed, as SUE converges to WE, it becomes impossible to meet the condition.

\citet{harsanyi1973games} argued that bounded rationality could also solve the instability problem he identified. By assuming each player's perception of other players' payoffs is subject to random errors, Harsanyi created a ``disturbed game" that is always stable because its equilibrium only admits pure strategies. An example of Harsanyi's disturbed game is our routing game based on SUE, in which every traveler chooses the route believed to be the best (i.e., a pure strategy). Here, we note that there is only one such route for a traveler because the probability of having two or more routes that are deemed the best by the traveler is zero when errors are continuous variables \cite[see][]{daganzo1977stochastic}.
As the random errors approach zero, the pure-strategy equilibrium of the disturbed game approaches the mixed-strategy equilibrium of the original game, and around the limit, the players would use their pure strategies approximately with the probabilities prescribed by the mixed-strategy equilibrium. However, this remedy, often known as Harsanyi's purification theorem, implies the mixed strategy equilibrium may only exist as an approximation to pure strategy equilibrium tied to exogenous random errors. Moreover, when errors are near zero, the difference in the payoffs between pure strategies diminishes, but the difference in the relative preferences for them, as manifested in the choice probabilities, may not. It is reasonable to expect a rational player to become increasingly indifferent to payoffs that become increasingly similar \citep{rosenthal1989bounded}. However, this behavior is not captured at or near the limit of the perturbed game. 

To recapitulate, the critiques on the perfect rationality assumption have generally led to an intellect exodus from WE. With bounded rationality, the definition of equilibrium is relaxed so that it can shift with the parameters chosen by the modeler to ``bound" rationality. If the goal is to match observations --- in terms of both individual route choices and aggregate traffic conditions --- such flexibility is no doubt a blessing. On the flip side, it can also be a curse to allow the equilibrium to depend on latent behavioral variables that could vary in space, time, and population. In the case of BRUE, the equilibrium is a set rather than a singleton, even at the aggregate level and with fixed behavioral variables. \citet{mahmassani1987boundedly} described this lack of uniqueness as ``the most disturbing question" since it ``poses a dilemma for flow prediction in networks." Thus, it is hardly surprising that WE remains widely used in practice as a reliable fall-back option for practitioners \citep{boyce2015forecasting}. {WE also provides a benchmark against which alternative equilibrium concepts based on bounded rationality can be evaluated, as it can often be viewed as their limit. This observation, however, lends no legitimacy to WE if one insists rationality must always be bounded  --- recalling that the above limit is precisely where the bound on rationality is gone. Nor does it guarantee the global stability of WE through a behaviorally sound dynamical process compatible with bounded rationality.}

\subsection{Our contribution}

Motivated by the above theoretical gaps, here we set out to show that a suitable behavioral theory of route choice can resolve the seemingly innate conflict between WE and bounded rationality. Under mild conditions, the proposed theory guarantees global stability. That is, boundely rational travelers can reach WE through a DTD dynamical process, \emph{regardless of} initial conditions. Moreover, travelers' route choices at WE, paradoxically, are compatible with both bounded and perfect rationality. In the parlance of game theory, this means the mixed-strategy equilibrium resulting from perfect rationality coincides with the probabilities of choosing pure strategies under bounded rationality. Therefore, the theory also solves  Harsanyi's instability problem. 

Our theory is built on a simple intuition: if two routes used at WE are assigned different choice probabilities, travelers must value these routes differently, even though their costs at WE are identical. To reconcile the ostensible contradiction in this statement, we conjecture that the route costs realized at WE are not the basis for deciding choice probabilities. Instead, travelers gradually build their valuation of each route through a DTD dynamical process. Consequently, the choice probabilities at WE reflect the \emph{preferences accumulated through the entire history} of that process, rather than just the experience at WE, which is achieved on the ``last day" (i.e., the limit) of the process.

In our theory, the cost experienced on a route each day \emph{accrues} to its valuation, whereas the DTD models in the literature typically view the valuation on a particular day as some average of the costs experienced up to that day \citep[e.g.,][]{horowitz1984stability, watling1999stability}. As we shall see, this is a subtle but vital difference in a setting with an infinite horizon. On each day, travelers act with bounded rationality, i.e., they assign a choice probability to each route based on the valuations accumulated hitherto; the better the valuation, the larger the probability. 
We shall prove this dynamical process converges to WE  for a rather broad class of cost accruement rules, including a naive addition rule (i.e., the valuation on day $k$ equals the sum of the costs experienced on day $t = 1, \cdots, k-1$). Our proof requires the choice probabilities to be determined by the logit model \citep{mcfadden1973conditional}, although other valuation-to-probability mappings may be considered as well. 

When our dynamical process reaches WE, travelers would still choose the used routes with probabilities mapped by their ``hidden" valuation of those routes. The benchmark route, which has the best valuation, receives the highest choice probability. Other routes are appraised against the benchmark. Travelers may be less inclined to use a route if it has a worse valuation than other routes, in accordance with the notion of bounded rationality.  They may also leave many routes unused. In our theory, these routes are interpreted as ``unacceptable," which, mathematically, means their valuation is unboundedly worse than the benchmark's.

While our contribution is largely theoretical, the proposed dynamical process does provide a prototype algorithm for solving a broad class of routing games. This practical value of global stability analysis has been recognized early in \citet{beckmann1956studies}.   An algorithm based on a DTD dynamical process is simple because it requires no more information than route travel costs to operate. Therefore, it can handle routing games with more general features, such as user heterogeneity and spatial interactions in travel costs.

\textbf{Organization.} The rest of the paper is organized as follows. We discuss related works in Section \ref{sec:related}. Section \ref{sec:setting} sets up the routing game and describes the WE and the basic DTD dynamical model. In Section \ref{sec:culo}, we present and interpret the cumulative logit ({CumLog}) model. In Section \ref{sec:stability}, we prove the global stability of the {CumLog} model. Section  \ref{sec:experiment} reports numerical experiments, and Section \ref{sec:conclusion} concludes the study.

\textbf{Notation.} We use $\sR$ and $\sR_+$ to denote the set real numbers and non-negative real numbers and use $\bar \sR = \sR \cup \{ \infty, -\infty\}$ to denote the set of extended real numbers.
For a vector $\va \in \sR^n$, we denote $\|a\|_p$ as its $\ell_p$ norm and denote $\supp{(\va)} = \{i: \eva_i > 0\}$ as its support.
For a matrix $\mA \in \sR^{n \times m}$, we denote $\|\mA\|_p$ as its matrix norm induced by the vector $\ell_p$ norm.
For two vectors $\va, \vb \in \sR^n$, their inner product is denoted as $\langle \va, \vb \rangle$. For a finite set $\sA$, we write $|\sA|$ as the number of elements in $\sA$ and $2^{\sA}$ as the set of all subsets of $\sA$.

\section{Related Studies}
\label{sec:related}

Our work focuses on the stability analysis of Wardrop equilibrium (WE) under the assumption of bounded rationality. 
In this section, we review the works that consider bounded rationality (Section \ref{sec:bounded-rationality}) and equilibrium stability (Section \ref{sec:dtd}) in game theory and transportation. Given the immensity of the literature that touches upon these topics, we limit our attention to those that are directly related to our work. 

\subsection{Bounded rationality}
\label{sec:bounded-rationality}

Bounded rationality is the idea that human decisions are affected by ``the knowledge that decision-makers do and don't have of the world, their ability or inability to evoke that knowledge when it is relevant, to work out the consequences of their actions, to conjure up possible courses of action, to cope with uncertainty, and to adjudicate among their many competing wants"  \citep{simon2000bounded}. This concept has been researched extensively by game theoreticians (see Section \ref{sec:brec}) and transportation researchers (see Section \ref{sec:brts}). 

\subsubsection{Application in games}
\label{sec:brec}
Vickrey's auction game \citep{vickrey1961counterspeculation}, Harsanyi's Bayesian game \citep{harsanyi1968games} and disturbed game \citep{harsanyi1973games} are earlier examples of games in which players are only boundedly rational, in the sense that they must deal with incomplete or imperfect information. \citet{selten1975reexamination}'s $\varepsilon$-perfect equilibrium and \citet{myerson1978refinements}'s $\varepsilon$-proper equilibrium also assume bounded rationality because they allow for the possibility that players choose sub-optimal strategies. \citet{van1987stability} examined why a player may mistakenly choose sub-optimal strategies.   He hypothesized that making a mental effort could help avoid such mistakes at the expense of a so-called ``control cost." The trade-off between finding the optimal strategy and minimizing this effort leads to a new game with bounded rationality. %

In an attempt to resolve the instability problem identified by \citet{harsanyi1973games} (see Section \ref{sec:instability}), \citet{rosenthal1989bounded} suggested another boundedly rational alternative to the standard game script. Rather than assuming players choose the best strategy with a probability of one (which implies perfect rationality), he argued that it is sufficient if equally good strategies are played with equal probabilities and better strategies are ``played with probabilities not lower than worse strategies."  This idea was further developed by \citet{mckelvey1995quantal} into the quantal response equilibrium (QRE, also known as boundedly rational Nash equilibrium) model, which essentially assigns choice probabilities to strategies based on the random utility theory \citep{mcfadden1973conditional}. Since the 1990s, the QRE game has been extended to deal with --- among other things --- extensive-form games \citep{mckelvey1998quantal}, auction games \citep{goeree2002quantal}, capacity allocation games \citep{chen2012modeling},  as well as Markov games \citep{chen2022adaptive}. %

The idea that players adopt inferior strategies with positive probability may also be viewed as a trade-off between exploration (gathering new information from uncharted territory) and exploitation (making the best use of information available). \citet{bjornerstedt1994nash} argued that players may need to try inferior strategies  in order to ensure they are indeed sub-optimal. They proposed an imitative dynamical process (more on this in Section \ref{sec:instability}) that allows players to use demonstrably sup-optimal strategies throughout the process, not because of mistakes or imperfection but because of the need for exploration. We note in passing that the exploration-exploitation trade-off is central to many machine learning (ML) algorithms, particularly bandit algorithms \citep{lattimore2020bandit} and reinforcement learning algorithms \citep{sutton2018reinforcement}. {Classical exploration strategies in ML include the random selection strategy --- selecting every strategy with at least a small probability, commonly known as ``$\varepsilon$-greedy" --- and the Boltzmann exploration strategy \citep{kocsis2006bandit}, which assumes the probability of each pure strategy to be selected is proportional to its exponential cost (mathematically, it is equivalent to the logit choice model). In non-cooperative games, if the players gradually weigh less toward exploration than exploitation, the learning process may be guided toward Nash equilibrium \citep{heliou2017learning}.}
 
\subsubsection{Application in transportation}
\label{sec:brts}

Dial's STOCH algorithm \citep{dial1971probabilistic} was probably the first attempt to replace traffic assignment based on WE with something that recognizes ``the non-optimal behavior of trip-makers" (i.e., bounded rationality). For what it was designed to do, i.e., performing a logit-based loading in an acyclic network, the algorithm was a remarkable success. However, \citet{dial1971probabilistic} did not conceive an alternative equilibrium concept. More importantly, when applied in traffic assignment, the logit model's reliance on the independence of irrelevance alternative (IIA) assumption can produce nonsensical results \citep{florian1976probabilistic}.  \citet{daganzo1977stochastic} proposed eliminating the IIA dependence by subjecting travelers to a normally distributed perception error on each link. This leads to the so-called probit model that is free of most problematic predictions of its logit counterpart but is much more computationally demanding
\citep[it usually requires Monte Carlo simulation, see, e.g.,][]{sheffi1981comparison}. The high computational cost of the probit model has motivated many to seek remedies within the logit framework. Most efforts aim to account for cross-route correlations, such as C-logit \citep{cascetta1996modified},  path-size logit \citep{ben1999discrete}, and generalized nested logit \citep{wen2001generalized}.  

\citet{daganzo1977stochastic} also introduced a boundedly rational version of WE, called stochastic user equilibrium (SUE), at which ``no traveler \emph{believes} he can improve his travel time by unilaterally changing routes." Clearly, bounded rationality here refers to travelers' inability to receive (or perceive) accurate information. \citet{fisk1980some} and \citet{sheffi1981comparison}  established equivalent mathematical formulations for SUE problems, respectively, based on the logit model and the probit model.
Conceptually, the logit-based SUE model is quite similar to the QRE model discussed in the previous section (though the QRE model was originally developed for $n$-person games), see \citet[][Section 4.2.1]{di2016boundedly} for a comparison. 

Bounded rationality may also be incorporated through Simon's satisficing theory \citep{simon1955behavioral}.  In the context of morning commute, \citet{mahmassani1987boundedly} introduced boundedly rational user equilibrium, or BRUE, which is attained when all travelers are satisfied with their choices, i.e., the gap between their current and optimal cost is within an \emph{indifference band} that reflects their aspiration level. They established the conditions for the existence of a BRUE and highlighted the non-uniqueness of such equilibrium.
\citet{hu1997day} incorporated indifference bands of tolerable ``schedule delay" into a simulation-assignment model to study the day-to-day evolution of network flows under real-time information and reactive signal control. Using data collected from a virtual laboratory experiment, \citet{mahmassani1999dynamics} confirmed the existence of the indifference band --- that is, travelers would not switch routes unless the improvement in trip time exceeds a certain threshold. \citet{mahmassani2000transferring} took the above virtual experiment approach one step further. They compared the findings from such experiments to those obtained from field surveys to determine the transferability of the insights. The fact that BRUE is not unique has inspired studies that attempt to characterize the BRUE set \citep{di2013boundedly} or to build an equilibrium selection model \citep{lou2010robust}. \citet{han2015formulation} formulated a BRUE problem that considers within-day dynamics (i.e., including both departure time and route choices) as a variational inequality problem and proposed several solution algorithms.

\subsection{Stability}
\label{sec:dtd}

The concept of stability is front and center in equilibrium analysis since equilibrium may be short-lived and difficult to reach without stability, thereby rendering it a useless construct. The stability of an equilibrium can be tested by the following questions: (i) can the equilibrium be restored after small perturbations (local stability), (ii) can the equilibrium be reached from any initial position (global stability), and (iii) can agents deviate from the equilibrium without penalty (Harsanyi's instability)? We shall focus on questions (ii) and (iii) above in this section (note that global stability implies local stability). In Sections \ref{sec:sta-evo} and \ref{sec:sta-dtd}, we review, respectively, classical dynamical models for games and in transportation, which were developed largely to answer the question of stability. Section \ref{sec:instability} deals with \citet{harsanyi1973games}'s instability.

\subsubsection{Learning and evolution in games}
\label{sec:sta-evo}

\citet{brown1951iterative} proposed an iterative process, called fictitious play, for solving certain finite games. His method assumes a player in each round simply responds to what they have ``learned" about the other player's strategy, represented as the empirical frequency of plays in the previous rounds. This is often viewed as the origin of the learning-based dynamical methods in games \citep{fudenberg1991game}.  The convergence of a fictitious play to mixed equilibrium was established by \citet{robinson1951iterative} for 2-person zero-sum finite games and by \citet{miyasawa1961convergence} for 2-person general-sum finite games with two pure strategies. However, \citet{shapley1964some} showed fictitious play could not ensure convergence in general 2-person games, thereby casting doubts on the global stability of mixed equilibrium of finite games. For finite games with bounded rationality, the stability of equilibrium is easier to establish. For example, \citet{fudenberg1993learning} proved that the equilibrium of the disturbed game studied in \citet{harsanyi1973games} can be reached through a learning-based dynamical process. Similarly,  \citet{chen1997boundedly} established conditions under which players can reach the quantal response equilibrium (QRE) of \citet{mckelvey1995quantal} through fictitious play. 

\citet{bush1955stochastic} suggested the strategies that have worked well in the past tend to be used more, as the positive experience is ``reinforced" through learning  \citep{cross1973stochastic}. In reinforcement learning (RL),  players need not form beliefs about others' strategies; instead, they simply update their strategies in response to realized rewards. RL algorithms may be linked to human behaviors in two ways \citep[see][for more details]{camerer2011behavioral}.  The first assumes the reward received by each player directly affects their future probability of choosing the same action: the higher the current reward, the greater the increases in the choice probability in the future \citep[e.g.,][]{cross1973stochastic, arthur1991designing}. The other interpretation posits a player's probability of selecting each action is determined by a ``score" associated with the action: the higher the current reward, the larger the increment in the score \citep{camerer1999experience, erev1998predicting}. %

Another line of thinking in the stability analysis for games originated from emulating biological evolution. Using game theory, \citet{smith1982evolution} argued that seemingly counter-intuitive behaviors (e.g., cooperation and altruism) can evolve and persist in a population because they are evolutionarily stable strategies. Since the theory applies to non-human species \citep{smith1973logic}, its validity does not rely on any form of human rationality (it is difficult to imagine ants as utility-maximizing creatures). The imitative dynamical process proposed by \citet{bjornerstedt1994nash} was an early application of evolutionary mechanisms --- selection, mutation, and replication --- in human competitions. It builds on a simple assumption: players tend to imitate the successful behavior of others. Specifically, players in each round switch from their current strategy $a$ to a pure strategy $b$  with a certain transmission probability, which increases with the utility of $b$ as well as the number of players selecting $b$ in the last round. Hence, a pure strategy is more attractive if it is not only more profitable but also more popular. \citet{bjornerstedt1994nash} proved their imitative dynamical process is locally stable. More recently, \citet{li2022differentiable} proved the global convergence. On the one hand, the imitative process differs from fictitious play in that it implies bounded rationality. On the other hand, unlike RL, it allows players to actively learn about and act on others' strategies. The reader is referred to \citet{weibull1997evolutionary} and \citet{sandholm2010population} for details on the evolutionary game theory. 

\subsubsection{Dynamical models in routing games}
\label{sec:sta-dtd}

The study of the route choice adjustment process, referred to as dynamical models in transportation, can be traced back to the stability analysis of WE by \citet{beckmann1956studies}. Most dynamical models operate on one of the following two mechanisms: (i) a discrete-time mechanism that maps travelers' valuation of available routes to route choice in discrete decision epochs, each representing one round of the routing game \citep{horowitz1984stability}. As the epoch is naturally a day in transportation, these models are often referred to as day-to-day, or DTD, models; (ii) a continuous-time mechanism in which the decision epoch is reduced to zero so that the relation between the change rate of route flows and the current route costs may be represented as an ordinary differential equation \citep{smith1984stability}. Our model falls into the first category. The reader may consult \citet{watling2003dynamics, cantarella2019dynamics} for a comprehensive review of dynamical models.

At the core of the discrete-time mechanism is modeling how travelers reevaluate and switch routes according to past experiences. As such learning processes do not involve anticipating other players' strategies, it is largely driven by reinforcement. \citet{horowitz1984stability} proposed that, on a given day, travelers may value a route based on a weighted average of either all experienced costs on that route before that day or of the cost and the valuation on the previous day. He showed that the global stability of the DTD process, even when equilibrium exists and is unique,  depends on how travelers incorporate past experience into the present route valuation. Instability ensues when the rate of adjustment is not properly selected (e.g., too much weight is given to either the recent past or the distant past). Horowtiz's schemes have since been extensively studied in transportation, with most efforts centering on tuning the weights in route evaluation, enriching behavioral contents, or establishing the existence and stability of equilibrium \citep[e.g.,][]{cascetta1989stochastic, cascetta1993modelling, cantarella1995dynamic,watling1999stability, watling2003dynamics,cantarella2016modelling}.

Because traffic conditions are subject to perception errors in Horowitz's model, his DTD process presumably converges to SUE. He did suggest the model may be employed to study WE when perception errors vanish but concluded his stability results could not be readily extended to the deterministic case. Indeed, while it is well known the fixed point of Horowitz's DTD dynamical process is SUE \citep[see, e.g.,][Section 3]{watling2003dynamics}, establishing its convergence to WE remains elusive even at the limit (i.e., when errors become zero). The primary difficulty, as noted in \citet{watling2003dynamics}, is that a discrete choice model without additional modeling devices cannot properly distribute travelers to a set of minimal and equal-cost routes according to the proportions prescribed by WE. 

By shrinking the decision epoch to zero, the continuous-time models center on moving flows between routes based on current costs. Behaviorally, this aggregate change is explained as travelers switching to routes that promise to lower their current costs. For example, the scheme proposed by \citet{smith1984stability} moves flow between every pair of routes at a rate proportional to the product of the flow on the higher-cost route and the cost difference. Using the Lyapunov theory, he proved this simple scheme leads to a globally stable dynamical system (i.e., it produces a solution trajectory converging to WE regardless of the initial solution) if the route cost function is monotone. Similarly, the Brown--von Neumann--Nash (BNN) scheme employed by \citet{yang2005evolutionary} shifts flow away from routes with travel costs above the weighted average of all routes. He also offered a behavioral explanation by interpreting the parameters in the BNN scheme as what \citet{cantarella1995dynamic} would call switching choice probability and route choice probability. Another widely used scheme is the so-called projected dynamical system, which may be viewed as a continuous-time version of the projection algorithm for variational inequality problems \citep{dupuis1993dynamical}. The idea is to change the route flows at a rate equal to the projection of the negative route cost vector onto the feasible set.   \citet{nagurney1997projected} noted this movement can be driven by ``travelers' incentive to avoid more costly routes $\cdots$ so that the sum of the flows equal the travel demand," though the direct linkage between the projection operation and actual route switching behaviors is somewhat abstract. Projected dynamical processes were also employed to establish the stability of routing games with elastic demands  \citep{friesz1994day,zhang1996local}, based on similar conditions used in \citet{smith1984stability}. The common requirement for global stability is the monotonicity of the route cost function. It is worth noting the terminology used in these papers is asymptotic global stability, which equals stability defined in \citet{smith1984stability} and global stability concerned herein. \citet{yang2009day} showed that each of the above continuous-time processes is a \emph{rational behavior adjustment process} (RBAP), which means their flow-shifting scheme always leads to a strict reduction in the total cost with a sufficiently small step size. This observation gave rise to a class of continuous-time DTD models operating at the link level  \citep{he2010link,guo2011bounded,guo2015link,di2015submission}.  These models have fewer behavioral contents than their route-based counterparts, as their flow-shifting schemes usually rely on a target link flow pattern obtained from solving an optimization problem to meet the RBAP requirement. More recently, \citet{smith2016route} and \citet{xiao2019day} incorporated logit dynamics into continuous-time models and established their convergence to SUE. 

\subsubsection{Harsanyi's instability problem}
\label{sec:instability}

\citet{harsanyi1973games} noted that a mixed strategy Nash equilibrium of a finite game is inherently unstable because, at the equilibrium, players can switch among equally good strategies (any of the pure strategies contained in the mixed strategy, or their combinations) without penalty. If players cannot be compelled by their self-interest to \emph{always} follow the prescription of the mixed strategy, it is difficult to sustain the equilibrium. In their celebrated book on equilibrium selection (Section 1.6), \citet{harsanyi1988general} named this  \emph{instability problem} one of the main difficulties with the concept of equilibrium in game theory.  We note that WE is affected by Harsanyi's instability problem as it is also a mixed strategy equilibrium of a finite game.

A common remedy to the instability problem is bounded rationality, which typically means introducing random errors into payoffs (travel costs). Examples include \citet{harsanyi1973games}'s disturbed game, \citet{daganzo1977stochastic}'s SUE model, and \citet{mckelvey1995quantal}'s QRE model.  Random errors suppress Harsanyi's instability problem because they reduce the probability of having two pure strategies with identical costs to zero. Strictly speaking, this approach does not fix the instability problem in mixed strategy equilibrium. It only posits that the existence of such equilibrium may be justified as an approximation to pure strategy equilibrium of the perturbed models. 

Another remedy is to assume players would never switch to an equally good strategy. \citet{bjornerstedt1994nash} argued that this assumption is implicit in Nash's prescription. In the transportation literature, the assumption has been widely used to develop dynamical models; see, for instance, \citet{smith1984stability}'s rule that allows flow shifting to occur between two routes only when their costs are strictly different.   Behaviorally, this may be explained as inertia, or the tendency to settle with one's current choices, especially when further search promises no additional benefits. However, inertia implies travelers would never explore inferior routes. If travelers do, as assumed in most models, base their choices on what they have learned from past experience, ruling out exploration altogether seems a strong assumption that is necessary only because otherwise, Nash equilibrium (or WE) would be cursed with Harsanyi's instability. Is the inertia assumption necessary? According to \citet{bjornerstedt1994nash}, the answer is no. In their imitative dynamical process, an equally good strategy adopted by more players provides more ``successful samples" for other players to imitate, thus attracting more players in the next round. Under this mechanism, players would continue to switch between equally good strategies even at equilibrium, but the imitative dynamical process ensures these movements do not push the system away from equilibrium. Thus, it resolves Harsanyi's instability problem without resorting to the inertia assumption or perturbation-based approximation.
 
\subsection{Summary}

To summarize what was reviewed earlier, bounded rationality may come from a decision maker's (i) inability to access accurate information (e.g., perception error), (ii) content with a sub-optimal choice compatible with their level of aspiration, (iii) erroneous choices and effort to avoid them, or (iv) desire to explore seemingly sub-optimal choices. In this paper, we interpret bounded rationality as imperfect choices --- in the sense that decision-makers allow themselves to use sub-optimal strategies based on their valuation --- which may be explained by any of the above four behavioral sources.

Our reading of the literature did not uncover a boundedly rational, behavior-driven, and globally stable dynamical process that can converge to WE, though such processes do exist for SUE. The convergence of continuous-time dynamical models to WE is well known. However, these models depend on perfect rationality and highly simplified learning and choice behaviors. They also need the assumption of inertia to overcome Harsanyi's instability problem. 

The dynamical process proposed herein precisely fills this gap. On the one hand, it explicitly incorporates bounded rationality, learning behaviors, and individual choices. On the other hand, it always converges, for any given initial point, to a  WE. Like the imitative dynamical process \citep{bjornerstedt1994nash}, our process also achieves immunity to Harsanyi's instability problem without assuming inertia. Unlike imitative dynamics, however, we do not assume travelers know the flows on each route (i.e., the basis for imitation), which is not public information in the context of routing games.

\section{Problem Setting}
\label{sec:setting}

A routing game takes place on a transportation network modeled as a directed graph $\gG(\sN, \sE)$, where $\sN$ and $\sE$ are the set of nodes and links, respectively. Let $\sW \subseteq \sN \times \sN$ be the set of OD pairs, and $\sK \subseteq 2^{\sE}$ be the set of available routes connecting all OD pairs. We use $\sK_w \subseteq \sK$ to denote the set of routes connecting $w\in\sW$ and $\sE_k \subseteq \sE$ the set of all links on route $k \in \sK$. Also, let $\emSigma_{w,k}$ be the OD-route incidence with $\emSigma_{w,k} = 1$ if the route $k \in \sK_w$ and 0 otherwise; and $\emLambda_{e,k}$ be the link-route incidence, with $\emLambda_{e,k} = 1$ if $e \in \sE_k$ and 0 otherwise. We write $\mLambda = (\emLambda_{e,k})_{e \in \sE, k \in \sK}$ and $\mSigma = (\emSigma_{w,k})_{w \in \sW, k \in \sK}$. Let $\vd = (\evd_w)_{w \in \sW}$ be a vector with $\evd_w$ denoting the number of travelers between $w \in \sW$. All travelers are identical, and their route choice strategies are represented by a vector $\vp = (\evp_k)_{k \in \sK}$, where $\evp_k$ equals the \textit{probability} that they select $k\in \sK_w$. The feasible region for $\vp$ can then be written as $\sP = \{\vp \in \sR_+^{|\sK|}: \mSigma \vp = \vone\}$.  The equilibrium of the routing game is characterized as travelers adopting a mixed strategy $\vp$ that minimizes their own travel costs. To simplify the discussion, we assume travelers between the OD pair adopt the same mixed strategy. According to the law of large numbers, $\vp$ can hence be equivalently viewed as the proportion of travelers selecting each route.  Let $\vf = (\evf_k)_{k \in \sK}$ and $\vx = (\evx_e)_{e \in \sE}$, with $\evf_k$ and $\evx_e$ being the flow (i.e., number of travelers) on route $k$ and link $e$, respectively. It follows $\vf = \diag(\mSigma^{\T} \vd) \vp$ and $\mLambda \vf = \vx$. Further define $\vu = (\evu_e)_{e \in \sE}$ as a vector of link cost, determined by a continuously differentiable function $u(\vx) = (u_a(\vx))_{a \in \sE}$ {(our analysis in the following sections does not require $\nabla u(\vx)$ be a diagonal or symmetric matrix)}. 
Then, the vector of route cost  $\vc = \mLambda^{\T} \vu$. To summarize, the route cost function $c: \sP \to \sR^{|\sK|}$ can be defined as $c(\vp) = \mLambda^{\T} \vu = \mLambda^{\T} u(\mLambda \vf) = \mLambda^{\T} u(\bar \mLambda \vp)$, where $\bar \mLambda = \mLambda \diag(\mSigma^{\T} \vd)$. 

\subsection{Wardrop equilibrium}

A Wardrop equilibrium (WE, \citet{wardrop1952road}) of the routing game can be defined as follows.
\begin{definition}[Wardrop equilibrium]
A route choice strategy $\vp^* \in \sP$ is a WE strategy if
$c_k(\vp^*) > \min_{k' \in \sK_w} c_{k'}(\vp^*)$ implies $\evp_k^* = 0$ for all $w \in \sW$ and $k \in \sK_w$.
\end{definition}
In other words, a route included in a WE strategy must have the minimum cost. It is widely accepted that travelers must be perfectly rational to reach and keep a WE strategy, which in our context means they always know the precise values of all route costs and consistently make correct choices accordingly \citep{sheffi1985urban}.

The WE routing game has an equivalent variational inequality problem (VIP) \citep{dafermos1980traffic}.
\begin{proposition}[VIP formulation of WE]
\label{prop:ue-vi}
    A route choice strategy $\vp^* \in \sP$ is a WE strategy if and only if it solves the following VIP: find $\vp^*\in \sP$ such that
    \begin{equation}
        \langle c(\vp^*), \vp - \vp^* \rangle \geq \vzero, \quad \forall \vp \in \sP.
        \label{eq:ue-vi}
    \end{equation}
\end{proposition}

We shall denote the solution set to the VIP \eqref{eq:ue-vi} as $\sP^*$, referred to as the WE strategy set. It is well known that $\sP^*$ is a singleton if  $c(\vp)$ is strongly monotone and a polyhedron if only $u(\vx)$ is \citep{dafermos1980traffic}. In the latter case,  although many WE strategies may exist, they must correspond to the same link flow $\vx^*$. The questions that concern us here are \emph{whether and how a WE strategy can always be achieved under reasonable assumptions of route choice behaviors}.  

Since \citet{beckmann1956studies}, many have asked these questions and have largely settled the ``whether" part. Specifically, it has been established that a WE strategy can be reached through a continuous-time dynamical process starting at any initial point, provided that the route cost function $c(\vp)$ is monotone \citep[see, e.g.,][]{smith1984stability,dupuis1993dynamical,friesz1994day,nagurney1997projected}. Yet, the answer to the question of ``how" is complicated by two issues. The {first} is behavioral. Note that travelers in the continuous-time models are supposed to be highly rational: they have accurate knowledge of and act on the most recent route costs to perpetually switch from higher-cost routes to lower-cost ones until a WE is reached. Therefore, continuous-time models have a limited capacity to accommodate such behaviors as learning from past experiences, exploring sub-optimal routes, and indifference to equal-cost routes. {Secondly} and perhaps more importantly, global stability is secured by implicitly applying an arbitrarily small rate of adjustment to the prescribed ``direction" of route flow changes \citep{watling1999stability}. However, whether the analysis is employed for the purpose of explaining real-world route-switching behaviors or of developing an equilibrium-finding algorithm, that rate cannot always be arbitrarily small. The questions are: what is a suitable magnitude of the adjustment at a given time,  how fast should this magnitude decrease as time proceeds, and how this pattern of time-varying adjustments is related to route choice behaviors? The discrete-time dynamical models are better equipped to address these questions. 

\subsection{Discrete-time dynamical model}
\label{sec:discrete-time}

At its core,  a discrete-time dynamical, or day-to-day (DTD), model keeps track of travelers' route valuation vector on day $t$, denoted as $\vs^t \in \sR^{|\sK|}$ ($t = 0, 1, \ldots$), which is mapped to their route choice strategy $\vp^t \in \sP$ by a function.  A commonly used route choice function is based on the logit model derived from random utility theory \citep{mcfadden1973conditional, ben1985discrete}.  Given 
a scalar $r > 0$, a logit-based  route choice function $q_r: \bar \sR^{|\sK|} \to \sP$ gives $\vp^t = q_r(\vs^t)$, where
\begin{equation}
    \evp_k^t =  \frac{\exp(-r \cdot \evs_k^t)}{\sum_{k' \in \sK_w} \exp(-r \cdot \evs_{k'}^t)}, \quad \forall k \in \sK.
    \label{eq:logit-w}
\end{equation}

By manipulating how $\vs^t$ is constructed and updated, the DTD models can represent a wide range of learning and choice behaviors. Below, we briefly review two most popular models. 

\textbf{The weighted average  model} assumes $\vs^t$ be a weighted average of the costs received in the past:
\begin{equation}
    \vs^t = \sum_{i = 0}^{t - 1} \eta^{ti} \cdot c(\vp^{i}), \quad \text{with}~\sum_{i = 0}^{t - 1} \eta^{ti} = 1,
    \label{eq:horowitz}
\end{equation}
where $\eta^{ti} \geq 0$ weighs how the cost received on day $i$ ($0 \leq i \leq t - 1$) affects the valuation on day $t$. Thus, the entire history of past experiences is allowed to affect the present-day decision. 

\textbf{The successive average model}, as a simplification of the weighted-average model, sets
\begin{equation}
    \vs^t = (1 - \eta^t) \cdot \vs^{t - 1} + \eta^t \cdot c(\vp^{t - 1}),
    \label{eq:watling}
\end{equation}
where $\{\eta^t\in (0,1); t = 1, 2,\cdots\}$ is a sequence of constants. In this model, the past experience is condensed into yesterday's valuation. This decision mode imposes a much lower information burden on travelers as it claims no direct memory of the experience prior to yesterday.

In the literature, both models were initially discussed by \citet{horowitz1984stability}. The latter can be viewed as a special case of the former, noting that recursively applying Equation \eqref{eq:watling} yields
\begin{equation}
    \vs^t = \eta^t \cdot c(\vp^{t - 1}) + (1 - \eta^t) \cdot \eta^{t - 1} \cdot c(\vp^{t - 2}) + \cdots + \prod_{i = 2}^{t} (1 - \eta^i) \cdot  \eta^1 \cdot c(\vp^0).
\end{equation}

\citet{horowitz1984stability} assumed the travelers' perception of $\vs^t$ is subject to a random error $\vepsilon^t \in \sR^{|\sK|}$. He considered two possibilities for the distribution of $\vepsilon^t$: the first (Model 1) assumes the distribution of $\vepsilon^t$ is independent of $t$, whereas the second (Model 2) treats $\vepsilon^t$ as the sum of the perception errors in the past (hence a function of $t$). Model 1 is much easier to analyze because it allows us to treat the parameter $r$ in the logit model \eqref{eq:logit-w}  as a time-invariant constant. In this case, if $(\vp^t, \vs^t)$ converges to a fixed point $(\hat \vp, \hat \vs)$, then we have $ \hat \vs = c(\hat \vp)$ and 
$\hat \vp = q_r(\hat \vs)$, and hence
\begin{equation}
    \hat \vp = q_r(c(\hat \vp)).
    \label{eq:sue}
\end{equation}
Following \citet{daganzo1977stochastic}, a route choice $\vp^*$ satisfies Equation \eqref{eq:sue} is a stochastic user equilibrium (SUE). The global stability of SUE under the successive-average model has been extensively studied. For a two-link network, \citet{horowitz1984stability} analyzed the global stability of the dynamical model \eqref{eq:watling} under the assumptions that the perception error is nonzero and the cost function is both monotone and Lipschitz continuous. Here, we note that \citet{horowitz1984stability}'s original Lipschitz continuous assumption is imposed on a composite function that combines $c(\vp)$ and the distribution function of $\vepsilon^t$. This assumption can always be satisfied when $c(\vp)$ is Lipschitz continuous, and the perception errors are not reduced to zero (that is, the stochastic model is not degraded to a deterministic one). He proved the model is globally stable if (i) $\sum_{t = 1}^{\infty} \eta_t = \infty$, and (ii) $\eta_t$ becomes sufficiently small for a sufficiently large $t$. \citet{cantarella1995dynamic} extended Horowitz's analysis to general networks but limited the stability analysis to the case where $\eta^t$ is time-invariant. Like \cite{cascetta1993modelling}, they introduced the switching probability $\alpha$ to describe the likelihood a traveler would even consider route choice on a given day. Thus,  the model in \citet{horowitz1984stability} can be viewed as a special case of their model when $\alpha = 1$ on every day. The global stability of their model can be guaranteed when either $\alpha$ or $\eta$ is sufficiently small --- a small $\alpha$ may be interpreted as strong habitual inertia, while a sufficiently small $\eta$ is the same requirement as in \citet{horowitz1984stability}. \citet{watling1999stability} analyzed the local stability of the model with a constant $\eta$ and investigated the possibility of applying Lyapunov's theory to determine its domain of attractions.

It is well known that SUE --- the fixed point of Equation \eqref{eq:sue} --- can be made arbitrarily close to WE by letting $r \to \infty$ \citep{fisk1980some, erlander1998efficiency, mamun2011select}. However, it remains an open question whether the stability result of SUE is applicable to WE at that limit. As noted by \citet{watling2003dynamics} (see their Example 4), given a fixed $r > 0$,  the stability of the successive-average model requires $\eta$ to be within $(0, \bar \eta_r]$, where $\bar\eta_r$ decreases to 0 when $r \to \infty$. Thus, attempting to approximate WE with the fixed point of \eqref{eq:sue} by choosing an arbitrarily large $r$ is problematic, as the feasible range of $\eta$ needed for convergence vanishes at the limit.

Is it possible to design a scheme that coordinates the increase of $r$ and the decrease of $\eta$ so that the final stationary point is steered toward WE? To the best of our knowledge, this question has not received much attention in the literature. Even if such a coordinated scheme can be identified, a more important question is how to make sense of it \textit{behaviorally}. Specifically, why would the travelers couple the changes in $r$ and $\eta$ in such a way as prescribed by the stability analysis?

In the next section,  we shall propose an alternative to the classical DTD dynamical models that is deceptively simple at first glance but holds promise to answer the above questions.

\section{Cumulative Logit ({CumLog}) Model}
\label{sec:culo}

We are now ready to propose a new DTD dynamical system that is dubbed, for the reasons that will soon become clear, the cumulative logit ({CumLog}) model. In developing the {CumLog} model, we were inspired by the conjecture in \citet{beckmann1956studies} about the study of stability (the emphasis is ours):
\begin{quote}
    Through \emph{a simple and plausible model}, one can get a rough picture of the \textit{minimum of conditions} that must be met in order that the adjustment process should converge to equilibrium.
\end{quote} 
Indeed, the overarching goal of this study is to develop that simple and plausible model envisioned by Beckmann and his co-authors and to identify ``the minimum of conditions" that ensure the convergence of a dynamical adjustment process to WE.  

The {CumLog} model adopts the basic framework of the DTD model \eqref{eq:horowitz}. That is, on day $t$, travelers update their route valuation vector $\vs^t$  based on the route costs on day $t - 1$ and select a route choice strategy $\vp^t$ according to $\vs^t$. Unlike \eqref{eq:horowitz},  $\vs^t$ in the {CumLog} model is updated, starting from some $\vs^0 \in \sR^{|\sK|}$, as follows:
\begin{equation}
\label{eq:culo}
    \vs^t = \vs^{t - 1} + \eta^t \cdot c(\vp^{t - 1}),
\end{equation}
where $\{\eta^t, t =1,2,\cdots\}$ is a sequence of positive constants.  Moreover,  travelers determine their strategy on day $t$  according to the logit model by setting $\vp^t = q_r(\vs^t)$ according to Equation \eqref{eq:logit-w}. The simplest interpretation of {CumLog} is that travelers value routes based on the cost received and \emph{accumulated} over the entire history up to $t-1$. Indeed, in the special case of $\eta =1$,  travelers literally add up all received costs without ever discounting the experiences in the distant past. 

Upon noticing the suspicious similarity between schemes \eqref{eq:culo} and \eqref{eq:watling}, some readers may understandably question the plausibility of our central claim: that scheme \eqref{eq:culo} can somehow ensure convergence to WE under mild conditions, while, as widely asserted in the literature, scheme \eqref{eq:watling} cannot. Therefore, in what follows, we shall first explain why averaging and accumulating route costs are fundamentally different in the dynamical process (Section \ref{sec:difference}). Section \ref{sec:interpretation} provides a behavioral interpretation of the {CumLog} model. An illustrative example is given in Section \ref{sec:illustration}.

\subsection{Difference between average and accumulation}
\label{sec:difference}
Let us revisit the thought experiment used to demonstrate Harsanyi's instability problem in Section \ref{sec:intro}. Suppose a routing game in a two-route network converges to a WE strategy that assigns Routes 1 and 2 a choice probability of 1/3 and 2/3, respectively.   Figures \ref{fig:convergence}-(i) and (ii) depict, respectively, how the flows and costs on the two routes gradually reach the WE through a dynamic adjustment process. The details of the process need not concern us here. Suffice it to say that at the end of the process, the costs on both routes are identical, and the probability of choosing Route 1 becomes 1/3, which implies a third of the travelers end up using that route. 
\begin{figure}[ht]
    \centering
    \begin{subfigure}[b]{0.32\textwidth}
        \centering
        \includegraphics[height=0.6\columnwidth]{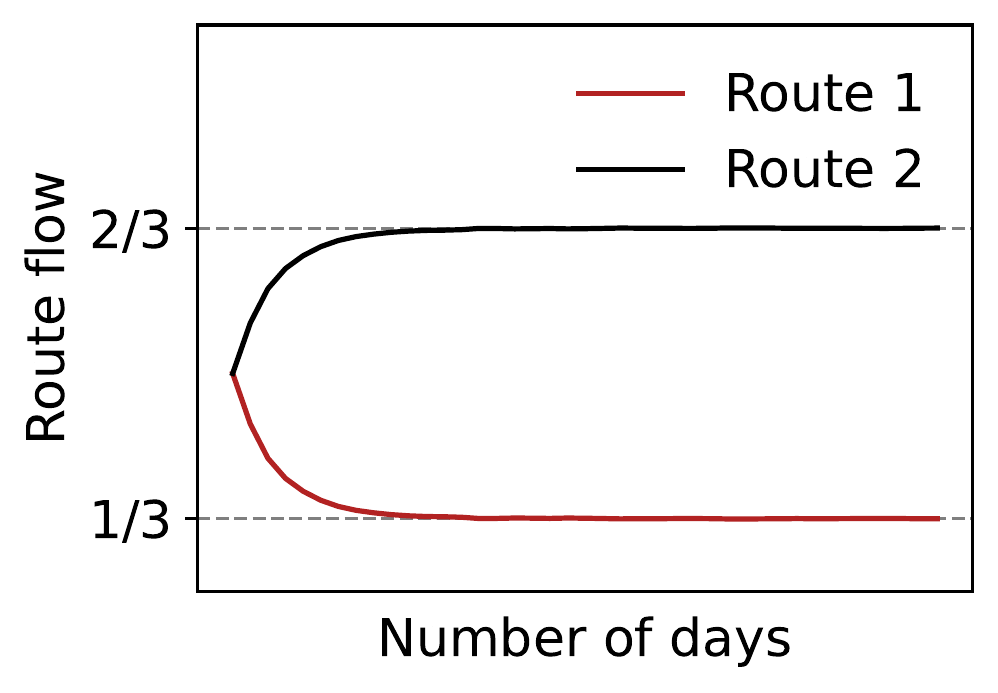}
        \caption{Evolution of flows.}
        \label{fig:flow}
    \end{subfigure}
    \begin{subfigure}[b]{0.32\textwidth}
        \centering
        \includegraphics[height=0.6\columnwidth]{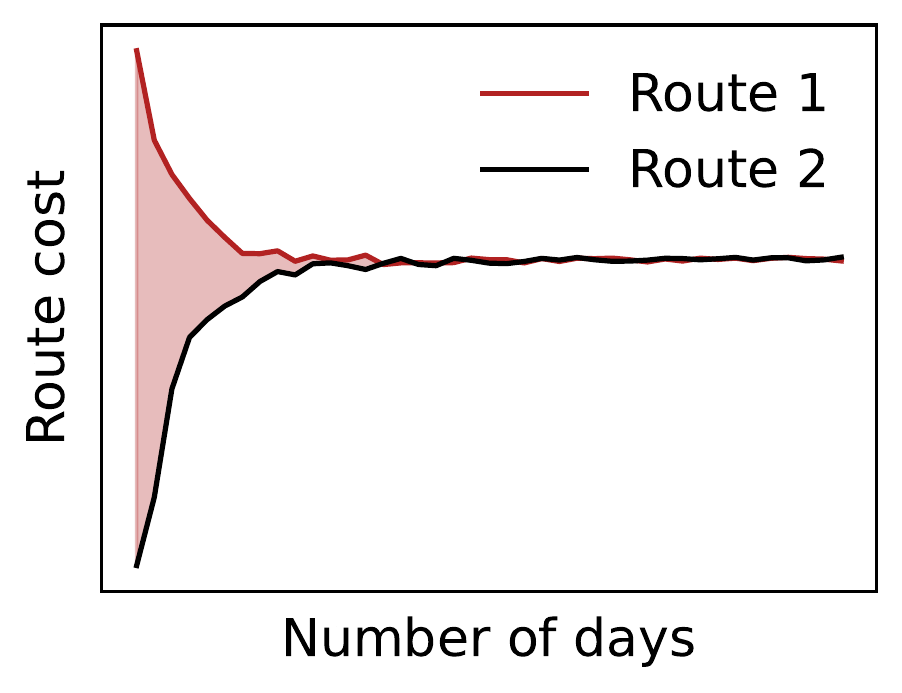}
        \caption{Evolution of costs.}
        \label{fig:cost}
    \end{subfigure}
    \begin{subfigure}[b]{0.32\textwidth}
        \centering
        \includegraphics[height=0.6\columnwidth]{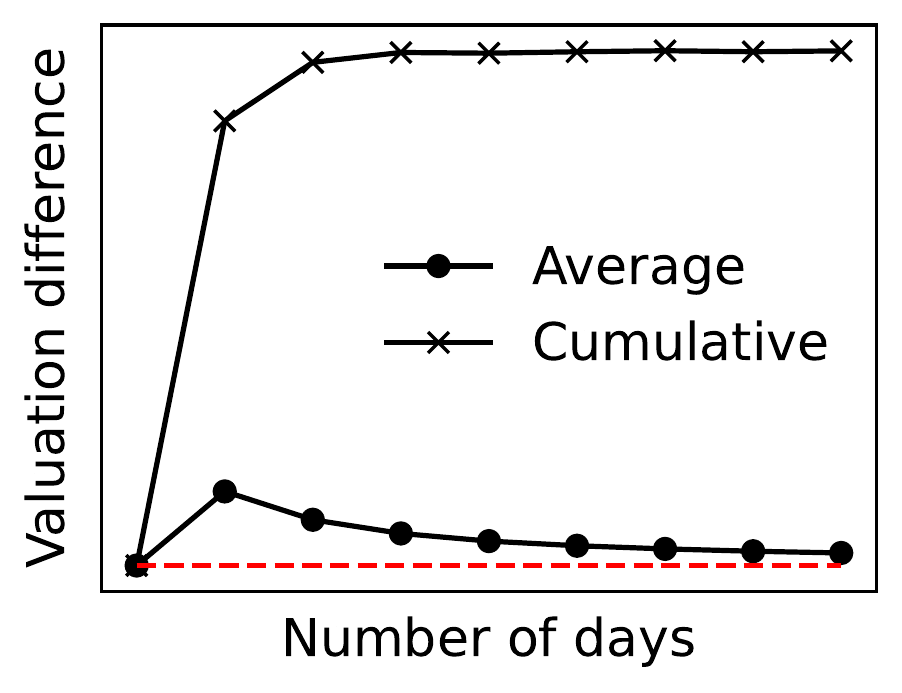}
        \caption{Valuation differences.}
        \label{fig:valuation}
    \end{subfigure}

    \caption{Illustration of a convergence process toward WE in a two-route routing game.}
    \label{fig:convergence}
\end{figure}

Figure \ref{fig:convergence}-(iii) compares the evolution of the difference in two route valuation schemes (average vs. accumulation) accompanying the dynamical process. In the average scheme, the difference in the route evaluations is bound to vanish regardless of the weights if the process ends at the WE. Given the same route valuations, it is difficult to see why at WE, travelers seem to prefer Route 2 over Route 1 with a 2:1 margin, as suggested by the mixed strategy.   Indeed, as \cite{harsanyi1973games} pointed out, there is no way to explain that preference other than insisting travelers prefer WE itself. This seemingly illogical preference is at the heart of Harsanyi's instability. Nor can this problem be explained away by bounded rationality. Note that the errors in the valuation would affect choice probability only when there is a non-zero difference in the ``deterministic" part of the valuations. At WE, the difference becomes zero. Hence, no errors could, on their own, swing travelers one way or the other. Unlike in the average scheme, the discrepancy in route valuations converges to a non-zero constant in the accumulation scheme (see the solid line with x markers in Figure \ref{fig:convergence}-(iii)). The constant equals the shaded area in Figure \ref{fig:convergence}-(ii), which visualizes the valuation difference accrued through the dynamical process. This cumulative difference then explains the mixed WE strategy: the travelers prefer Route 2 at equilibrium because they value it substantially (but not infinitely) more than Route 1, although the equilibrium route costs are the same.

\smallskip
\begin{remark}
    In behavioral economics, the reinforcement learning model proposed by \cite{erev1998predicting} assumes that players in a finite game have the highest propensity to choose the pure strategy that gives the greatest total reward in the past. This assumption is similar to our assumption that the route with a smaller cumulative route cost is selected with a greater probability.
\end{remark}

\subsection{Behavioral interpretation}
\label{sec:interpretation}

The route choice behaviors implied by the {CumLog} model can be summarized as follows.

\begin{itemize}
    \item {On each day, travelers choose \textit{afresh} a mixed route choice strategy based on current route valuations through a logit model. No additional assumptions are needed regarding the behavioral inertia --- the reluctance to make changes once a choice becomes habitual --- that is often explicitly modeled by a switching probability in the DTD literature \citep[e.g.,][]{cantarella1995dynamic}.
    \item Route choices are driven by relative valuation rather than absolute valuation, as dictated by the logit model. The benchmark is the ``best" route that receives the highest choice probability. Other routes are appraised against the benchmark: they shall also be selected with a strictly positive probability as long as the valuation difference between these routes and the benchmark is finite; they may be selected with a zero probability if their valuations are deemed unacceptable, i.e., \textit{infinitely} worse than the benchmark.}
\end{itemize}

{Here lays the rationale behind our claim that \emph{the CumLog model is boundedly rational}: rather than committing to never moving to a route with a worse valuation than they currently enjoy, the travelers consistently assign a non-zero probability to sub-optimal routes with an \emph{acceptable} valuation.

The two parameters of the CumLog model can be linked to route choice behaviors as follows.
\begin{itemize}
    \item {The parameter $r$ characterizes how travelers' route choice strategy $\vp^t$ is determined from the valuation vector $\vs^t$.} It measures the trade-off between exploration and exploitation: the larger the parameter $r$, the more exploitative the travelers (meaning they are less likely to explore sub-optimal routes). Thus, $r$ will be referred to as the \emph{{exploitation} parameter}. In this study, the parameter $r$ is fixed at a constant value.   One may interpret this setting in one of two ways:  (i) the perception errors are independent of $t$, which is the assumption used to justify Model 1 in \citet{horowitz1984stability}; or (ii)  travelers' propensity for accepting sub-optimal routes, or their desired balance between exploration and exploitation, is time-invariant \citep{fudenberg1993learning, kocsis2006bandit}. Our stability analysis is agnostic on the interpretation of the {exploitation} parameter.
    \item {The parameter $\eta^t$ regulates how the valuation vector $\vs^t$ is updated.}  Before the routing game is played, travelers have an initial route valuation vector $\vs^0$. If no prior preference exists, then they simply set $\evs_k^0 =0$ for all $k \in \sK$.
    On day $t\geq 1$, the travelers update the valuation $\vs^t$  by adding to it a cost vector $\eta^t \cdot c(\vp^{t-1})$ where $\eta^t\ge 0$ is the weight on day $t$.  The weight $\eta^t$ measures the impact of the cost received on day $t - 1$ on the travelers' valuation on day $t$. 
    Behaviorally, it captures how quickly travelers become disposed to ignore the latest information and ``settle down." Thus, $\eta^t$ will be referred to as the \emph{{proactivity} measure}. The larger the $\eta^t$, the less passive the travelers. As we shall see, the stability of the {CumLog} model depends on the asymptotic behavior of the {proactivity} measure  $\eta^t$.  Not all feasible sequences of $\eta^t$ guarantee convergence to WE. For example, if travelers stop incorporating new information into route valuation too soon (indicating a rapid descent to extreme passivity), {CumLog} may stabilize quickly but at a place far away from WE.   We shall consider two asymptotic rules for the {proactivity} measure in this study. In the first,  $\eta^t = \eta > 0$ for all $t \geq 0$, i.e.,  the level of {proactivity} remains at a constant level through the entire process. The second rule dictates that $\eta^t$ monotonically decreases to 0 as $t \to \infty$. Thus, the costs received by the travelers will have a progressively diminishing impact on their route choice. Another rule, in which $\eta^t$ converges to some constant $\eta >0$ as $t \to \infty$, may be inferred from the above two.  
\end{itemize}
}

\subsection{Illustrative example}
\label{sec:illustration}

We close by illustrating CumLog with a simple routing game played on a network with three parallel links connecting an OD pair. The cost functions on the links are $u_1(\evx_1) = \evx_1$, $u_2(\evx_2) = \evx_2 + 1$, and $u_3(\evx_3) = \evx_3 + 2.25$. The total demand is $d = 3$.   It can be easily verified the WE conditions dictate the three links be selected with probabilities $\evp_1^* = 2/3$, $\evp_2^* = 1/3$, and $\evp_3^* = 0$. The WE strategy is unique in this case because all route cost functions are strictly increasing. 

Setting $r = 0.25$ and fixing $\eta^t = 1$ and starting from $\vs^0 = 0$, we run the {CumLog} model from day 0 to day 12 and report the convergence process in Figure \ref{fig:example}. A WE is reached after day 12, with the proportion of travelers selecting each route converging rather precisely to $2/3$, $1/3$, and $0$, respectively, and the costs on the two routes included in the mixed strategy, Routes 1 and 2, become identical.  

\begin{figure}[ht]
    \centering
    \begin{subfigure}[b]{0.4\textwidth}
        \includegraphics[width=0.8\columnwidth]{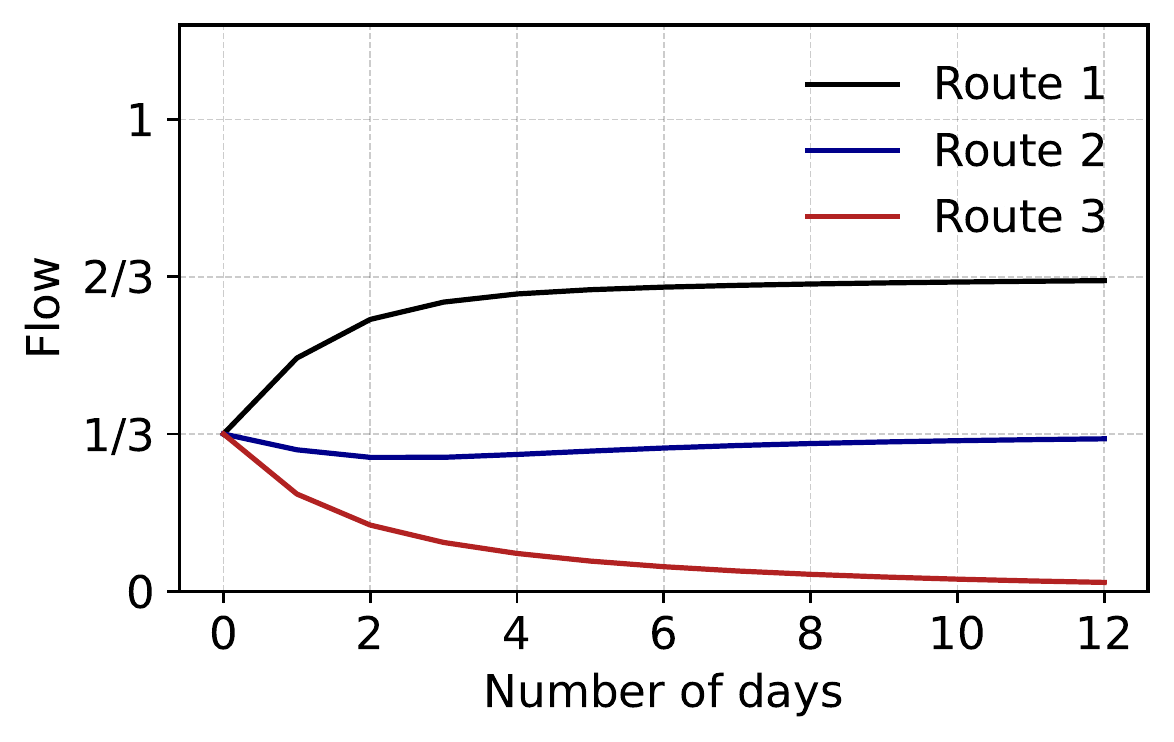}
        \caption{Route choices.}
        \label{fig:eg-flow}
    \end{subfigure}
    \begin{subfigure}[b]{0.4\textwidth}
        \includegraphics[width=0.8\columnwidth]{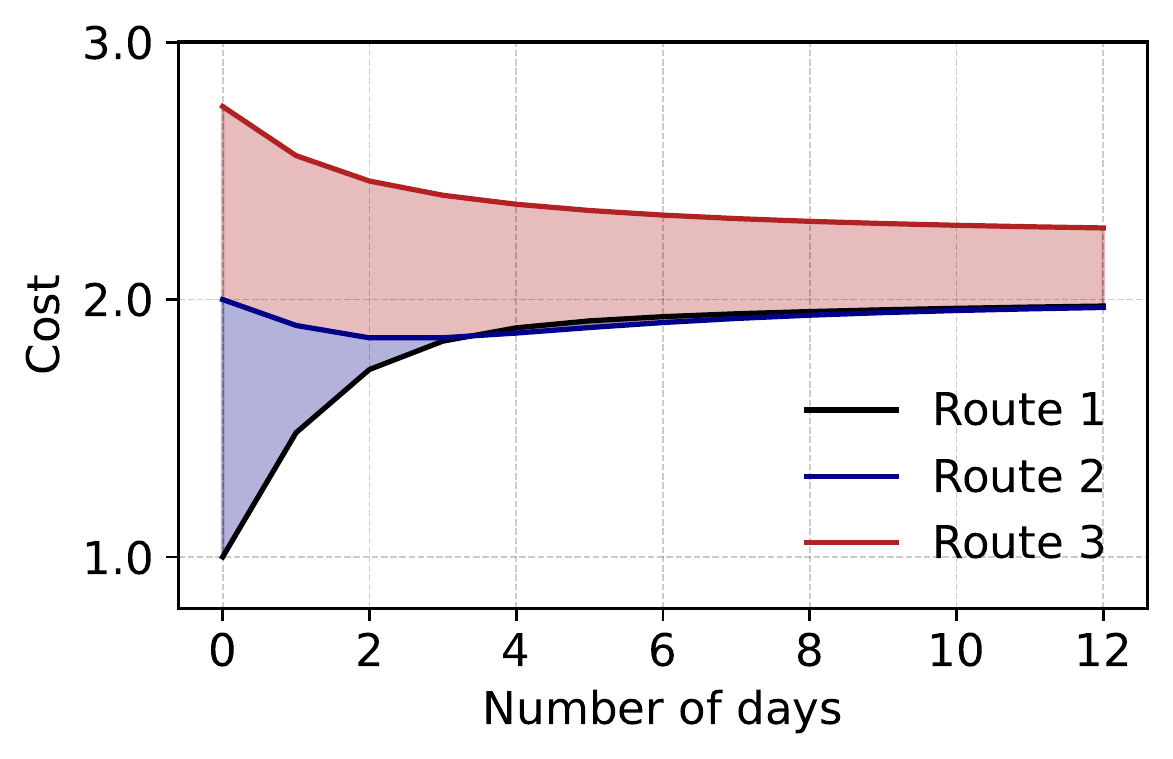}
        \caption{Travel costs.}
        \label{fig:eg-cost}
    \end{subfigure}
    \caption{Illustration of a convergence process toward WE in a three-route routing game.}
    \label{fig:example}
\end{figure}

In Figure \ref{fig:example}-(ii), the blue and red areas highlight the difference in valuation between Routes 1 and 2 and that between Routes 2 and 3, respectively. Route 1 is always the lowest-cost route throughout the process and thus is always selected by most travelers. The blue area approaches a constant value as $t$ increases. Consequently, the relative preference for Route 1 over Route 2 became stabilized, indicating Route 2 is an inferior but acceptable option. The red area, however, grew to infinity as $t\to \infty$, which means Route 3 became infinitely worse than Route 2 and eventually was abandoned. 

This example exhibits another distinction between the {CumLog} model and the classical  DTD models. Even with a finite {exploitation} parameter $r$, the {CumLog} model is capable of identifying and eliminating the routes that no WE strategy should use.   The classical models, however, are obliged by their averaging scheme to keep a positive flow on every route unless $r \to \infty$. This result can be expected from the fact that the limiting point of these models is SUE rather than WE.

\section{Global Stability}
\label{sec:stability}
In this section, we present and prove the main stability result concerning the {CumLog} model proposed in Section \ref{sec:culo}. Simply put, the objective is to show that, under mild conditions, the DTD dynamical model \eqref{eq:culo} always converges to a WE regardless of the initial point. The following assumptions describe some of the conditions {on the link cost function $u(\vx)$,  whose domain (the set of feasible link flows) is written as $\sX = \{\vx: \sR^{|\sA|}: \vx = \bar \mLambda \vp, \ \vp \in \sP\}$.}

\begin{assumption}
\label{ass:differentiable}
    The link cost function $u(\vx)$ is twice continuously differentiable {on $\sX$}.
\end{assumption}

\begin{assumption}
\label{ass:cocoercive}
    {For all $\vx \in \sX$, the symmetric part of both $\nabla u(\vx)$ and $(\nabla u(\vx))^2$ are positive semi-definite.}
\end{assumption}

Assumption \ref{ass:cocoercive} is satisfied as long as $u(\vx)$ is monotone and $\nabla u(\vx)$ is symmetric. If $\nabla u(\vx)$ is asymmetric, Assumption \ref{ass:cocoercive} still holds if $\nabla u(\vx)$ is not ``too asymmetric," i.e., the anti-symmetric part does not exceed the symmetric part \citep[see, e.g.,][for a more rigorous description]{hammond1987generalized}. As the assumptions require neither $u(\vx)$ to be strongly monotone nor $\nabla u(\vx)$ to be symmetric,  our analysis can be applied to a broad class of routing games, including those with non-separable and non-strictly increasing cost functions. 

The next proposition characterizes the route cost function $c(\vp)$ and the WE strategy set $\sP^*$. 
\begin{proposition}
\label{prop:property}
    Under Assumptions \ref{ass:differentiable}-\ref{ass:cocoercive}, (i) there exists $L \geq 0$ such that $c(\vp)$ is $1/4L$-cocoercive on $\sP$, i.e., $\langle c(\vp') - c(\vp), \vp' - \vp \rangle \geq 1/4L \cdot \|c(\vp') - c(\vp) \|_2^2$ for all $\vp, \vp' \in \sP$; (ii) $\langle c(\vp), \vp - \vp^*\rangle \geq 0$ for all $\vp \in \sP$ and $\vp^* \in \sP^*$, with equality for any given $\vp^* \in \sP^*$ if and only if $\vp \in \sP^*$; and (iii) $\sP^*$ is convex and compact. 
\end{proposition}

\begin{proof}

   Following the standard convex analysis techniques we can show Assumptions \ref{ass:differentiable}-\ref{ass:cocoercive} $\Rightarrow$ (1) $\Rightarrow$ (2)  $\Rightarrow$ (3). 
   The reader is referred to Appendix \ref{app:property} for details.
\end{proof}

We note that  when $\nabla c(\vp)$ is symmetric, $L$ is the Lipshcitz constant of $c(\vp)$ on $\sP$, i.e.,
any $L \geq \max_{\vp \in \sP} \|\nabla c(\vp)\|_2$ can be used to fulfill the requirement.  The case with asymmetric $\nabla c(\vp)$ requires $L = \max_{w \in \sW} \{\evd_w\} \cdot H \cdot \|\mLambda\|_2^2$,  where $H$ is the Lipshitz constant of $u(\vx)$ on $\sX = \{\vx \in \sR_+^{|\sA|}: \vx = \bar \mLambda \vp, \, \vp \in \sP\}$. 

\textbf{Main result.} We are now ready to give the main stability result.  
\begin{theorem}
\label{thm:convergence-ue}
    Under Assumptions \ref{ass:differentiable}-\ref{ass:cocoercive}, suppose that $\vs^0 < \infty$, then $\vp^t$ in the {CumLog} model \eqref{eq:culo} converges to a fixed point $\vp^* \in \sP^*$, the solution set to the VI problem \eqref{eq:ue-vi} if either of the following two conditions is satisfied: (i) $\lim_{t \to \infty} \eta^t = 0$ and $\lim_{t \to \infty} \sum_{i = 0}^t \eta^i = \infty$, or (ii) $\eta^t = \eta < 1/2rL$ for all $t \geq 0$.
\end{theorem}

Below, Section \ref{sec:stability-interpret} interprets {the convergence result}, while Section \ref{sec:stability-sketch} sketches the proof of Theorem \ref{thm:convergence-ue} (the complete proof is given in Appendix \ref{app:iteration} to \ref{app:convergence-ue}).  

\subsection{Interpretation of the {convergence result}}
\label{sec:stability-interpret}

\subsubsection{Convergence conditions} The two convergence conditions given by Theorem \ref{thm:convergence-ue} can be interpreted as follows.
\begin{itemize}
    \item The first part of Condition (i) in Theorem \ref{thm:convergence-ue} implies the attention paid to the latest route cost decreases with $t$ and eventually disappears altogether. In other words, travelers tend to settle down in the long run. The second part states that convergence may be at risk if travelers settle down too soon. Specifically, the decreasing rate of $\eta^t$ cannot be faster than $\mathcal{O}(t^{-1})$. The following are a few examples that meet this ``not-too-soon" requirement:  $\eta^t = \mathcal{O}(1/t)$, $\eta^t = \mathcal{O}(1/\sqrt{t})$ or $\eta^t = \mathcal{O}(1/\log(t))$.  The second part is introduced to ensure the routes not included in any WE strategy become infinitely worse than the best route when $t \to \infty$. If the condition is violated, say for example $\eta^t = \mathcal{O}(t^{-(1 + \delta)})$ for some $\delta > 0$, the valuation $\vs^t$ would remain bounded when $t \to \infty$. {The monotone convergence theorem then guarantees the convergence of $\vs^t$. Denoting its limit as $\bar \vs$,  we can show the difference between any elements of $\bar \vs$ is bounded since $\bar\vs$ itself is bounded.} As a result, per the logit model, all routes are bound to receive positive flow {as $t \to \infty$}, which, in general, violates the WE conditions.
    \item Condition (ii) means travelers would never stop incorporating new information into route valuations. Rather, their propensity for {proactivity} is maintained at a constant level below a certain threshold. Compared to (i), this is a weaker requirement, and thus, the convergence under it is more difficult to establish. Condition (ii) may also be interpreted as given a fixed level of {proactivity} $\eta$, the {exploitation} parameter $r$ must not exceed $1 / 2\eta L$. At first glance, this upper bound on $r$ is puzzling because one would expect a larger $r$ makes it easier to reach WE since, in theory, WE is the limiting case of SUE when $r \to \infty$. However, a moment of reflection reveals that, in a dynamical process, a large and constant $r$ means that travelers tend to under-explore the route space, which might prevent them from reaching WE. In learning theory, it is well known that a sufficient level of exploration in the early stage is critical to effective exploitation in the later stage \citep{lattimore2020bandit}. 
\end{itemize}

{
\subsubsection{Significance}

The above interpretation of the convergence conditions indicates that CumLog can reach Wardrop Equilibrium (WE) under relatively mild assumptions about boundedly rational route choice behaviors. This finding challenges the longstanding belief in transportation research that WE has a shaky behavioral foundation \citep[see e.g.,][]{watling2003dynamics}. 
}

{Our theory also resolves Harsanyi's instability problem, which manifests as follows. When WE is reached, all used routes are equally good. Hence, breaking the tie arbitrarily is {not} against anyone's interest. However,  if travelers do break the tie arbitrarily, the system cannot stay at WE. The CumLog model solves this dilemma by allowing travelers to assign a different valuation (hence a different choice probability) to routes that have identical costs at WE.  Specifically, when CumLog reaches a WE strategy $\vp^*$, given any OD pair $w \in \sW$ and two WE routes $k, k' \in \sK_w$ between that OD pair, travelers' valuation difference between Routes $k$ and $k'$ would be
$$- \frac{\log(\evp_k^*) - \log(\evp_{k'}^*)}{r}.$$
This valuation difference explains why two identical-cost routes are selected with different probabilities. }

\subsection{Sketch of the proof}
\label{sec:stability-sketch}

To sketch the proof of Theorem \ref{thm:convergence-ue}, let us define $\vp_w = \{\evp_k\}_{k \in \sK_w}$ and the negative entropy function as $\phi_w(\vp_w) = \langle \vp_w, \log(\vp_w) \rangle$. Also, define $\tilde \eta^t \equiv r \cdot \eta^t$ and $\sQ_w^t = \{\vp_w \in \sP_w: \supp(\vp_w) \subseteq \supp(\vp_w^t)\}$. We shall prove Theorem \ref{thm:convergence-ue} by showing that there exists $\vp^* \in \sP^*$ such that the distance between $\vp^*$ and $\vp^t$ converges to 0 when $t \to \infty$. Since the decision variable is defined on $\sP = \prod_{w \in \sW} \sP_w$, the Cartesian product of a set of probability simplex $\sP_w = \{\vp_w \in \sR_+^{|\sK_w|}: \vone^T \vp_w = 1\}$, the KL divergence (also known as the statistical distance) is a natural measure, given by
\begin{equation}
    D(\vp^*, \vp^t) = \sum_{w \in \sW} D_w(\vp_w^*, \vp_w^t) = \sum_{w \in \sW} \phi_w(\vp_w) - \phi_w(\vp_w^t) - \langle \nabla \phi_w(\vp_w^t), \vp_w - \vp_w^t \rangle.
\end{equation}

The proof is divided into three steps. 

\textbf{Step 1.} We first establish the relation between $\vp^{t + 1}$ and $\vp^t$ in the {CumLog} dynamical process.  
\begin{lemma}
\label{lm:iteration}
    Given any $r > 0$, $\eta^t > 0$, and $\vs^t \in \bar\sR^{|\sK|}$, let $\vp^t = q_r(\vs^t)$. Then for a vector $\tilde \vp\in \sP$, we have $\tilde \vp = q_r(\vs^t + \eta^t \cdot c( \vp^t))$ if and only if  
    \begin{equation}
    \begin{split}
        \langle \nabla \phi_w(\vp_w^t) -  \nabla \phi_w(\tilde \vp_w),  \vp_w - \tilde \vp_w \rangle \leq \tilde \eta^t \cdot  \langle c_w(\vp^t),  \vp_w - \tilde \vp_w\rangle,  \quad \forall \vp_w \in \sQ_w^t, \quad \forall w \in \sW.
        \label{eq:p-2}
    \end{split}
    \end{equation}
\end{lemma}

\begin{proof}
    The proof follows from two well-known results: (1) the equivalence between the logit model and a convex program (Proposition \ref{lm:projection-logit}); (2) the equivalence between the convex program and a VIP (Lemma \ref{eg:optimization}). 
    The reader is referred to  Appendix \ref{app:iteration} for the complete proof. 
\end{proof}

\textbf{Step 2.} We next link $D(\vp^*, \vp^{t + 1})$ to $D(\vp^*, \vp^{t})$.

\begin{lemma}
\label{lm:mid}
    If $\vs^0 < \infty$, then given any $\vp^* \in \sP^*$, $\vp^t$ and $\vp^{t + 1}$ in the {CumLog} dynamical process satisfy 
    \begin{equation}
        D(\vp^*, \vp^{t + 1}) \leq D(\vp^*, \vp^t) - \frac{1}{2} \cdot \|\vp^t - \vp^{t + 1} \|_2^2 + \tilde \eta^t \cdot  \langle c(\vp^t),  \vp^* - \vp^{t + 1}\rangle.
    \end{equation} 
\end{lemma}

\begin{proof}
    As $\vs^0 < \infty$, we have $\vs^t < \infty$ and thus $\sQ_w^t = \sP_w$ for all  $t \geq 0$. Invoking Lemma \ref{lm:iteration} then completes the proof.
    See Appendix \ref{app:mid} for details.
\end{proof}

\begin{proposition}
\label{prop:decreasing}
    Under Assumptions \ref{ass:differentiable}-\ref{ass:cocoercive}, if $\vs^0 < \infty$, then given any $\vp^* \in \sP^*$, $\vp^t$ and $\vp^{t + 1}$ in the {CumLog} dynamical process satisfy 
    \begin{equation}
        D(\vp^*,  \vp^{t + 1}) \leq D(\vp^*, \vp^t) - \frac{1 - 2\tilde \eta^t L}{2} \cdot \|\vp^t - \vp^{t + 1} \|_2^2.
    \end{equation}
\end{proposition}

\begin{proof}
   As per Proposition \ref{prop:property}, $c(\vp)$ is $1/4L$-cocoercive. The result then follows from Lemma \ref{lm:mid} and the cocoercivity of $c(\vp)$.
   See Appendix \ref{app:decreasing} for the complete proof. 
\end{proof}

\textbf{Step 3.}  We finally prove the convergence of $\vp^t$ toward a point in $\sP^*$ under Conditions (i) and (ii).

\begin{proof}
Proposition \ref{prop:decreasing} indicates that $\lim_{t \to \infty} D(\vp^*, \vp^t)$ exists for all $\vp^* \in \sP$ under both Conditions (i) and (ii). All left to prove is to find a convergent subsequence $\{\vp^{t_j}\} \subseteq \{\vp^t\}$ with $\vp^{t_j} \to \sP^*$ when $j \to \infty$. Denote $\hat \vp$ as the limit of $\vp^{t_j}$.  Then $\vp^{t_j} \to \hat \vp$ implies $D(\hat \vp, \vp^{t_j}) \to 0$. As $\lim_{t \to \infty} D(\hat \vp, \vp^t)$ exists, it follows $D(\hat \vp, \vp^t) \to 0$, which implies $\vp^t \to \hat \vp$. Under Condition (i) given in Theorem \ref{thm:convergence-ue}, we prove by contradiction: if there does \textit{not} exist a subsequence of $\{\vp^t\}$ that converges to $\sP^*$,  Properties (ii) and (iii) given in Proposition \ref{prop:property} cannot both hold. Under Condition (ii), we first invoke the Bolzano-Weierstrass theorem to exact a convergent subsequence and then prove its limit must belong to $\sP^*$.
The reader is referred to Appendix \ref{app:convergence-ue} for details.
\end{proof}

\smallskip
{
\begin{remark}[On infinite valuations]
    In our model, when either Condition (i) or (ii) is met, $\vs^t$ converges to infinity when $t \to \infty$.  While an infinite value might seem unrealistic, a key feature in the logit model is that the choice probability depends on the relative, rather than absolute, valuations of alternatives.  In our context, this may be interpreted as travelers monitoring the differences between the elements in the valuation vector $\vs^t$. Thus, as long as these differences are finite, the model will give correct results.   One can also modify the original model to avoid such a nuanced interpretation. To simplify the discussion, assume $|\sW| = 1$ (i.e., there is only one OD pair).  Consider the following two behavioral variants.
    \begin{itemize}
        \item \textbf{Variant 1}. On each day $t$, after updating route valuations, 
        travelers further adjust them such that the best-valued route is normalized to zero, i.e., replacing $\vs^t$ by $\vs^t - \min(\vs^t)$.  
        \item \textbf{Variant 2}. One each day $t$, the travelers update route valuations by $\vs^t = \vs^{t - 1} + \eta^t \cdot (c(\vp^{t - 1}) - \min(c(\vp^{t - 1}))$ so that the valuation of the route with the lowest cost on day $t$ remains unchanged.
    \end{itemize}
    \smallskip
    Mathematically, both variants are identical to the original CumLog model. However, they would ensure travelers' valuation of routes used at WE be bounded. However, the valuation on the non-WE routes would still grow to infinity, signifying they are unacceptable. 
\end{remark}
}

\smallskip
{
\begin{remark}[Related convergence results]
The two conditions given in Theorem \ref{thm:convergence-ue} are similar to the conditions detailed in Theorem 1 of \citet{horowitz1984stability}. The difference is twofold. First, his proof is established only for the two-link network. Second and most important, his dynamical system, based on an average rather than a cumulative scheme, converges to SUE rather than a WE. \citet{horowitz1984stability} pondered on the possibility of extending his stability result to the WE case (what he called ``the deterministic model"). His conclusion was negative because letting $r \to \infty$ (the equivalent of zero perception errors) not only violates Lipshcitz continuity but also voids the argument that the convergence of the average cost difference to the equilibrium value must imply the convergence to the correct equilibrium (see his Example 4). 

The {CumLog} model also shares a similar mathematical structure with a classic online convex programming method known as the dual averaging (DA) algorithm \citep{nesterov2009primal, xiao2009dual}. Recently, \citet{mertikopoulos2019learning} studied the convergence of the DA algorithm in continuous games. The convergence conditions given by them are similar to Condition (i) in Theorem \ref{thm:convergence-ue}, though they were proved using a  different technique.  Condition (ii) in Theorem \ref{thm:convergence-ue}, to the best of our knowledge, has never been rigorously shown to ensure the convergence of the DA algorithm in games or VIPs.

\end{remark}
}

\smallskip
{
\begin{remark}[On the convergence of the successive average model to WE]
\label{rm:eta-r}
    At the end of Section \ref{sec:discrete-time}, we ask whether properly coupling the increase of $r$ and the decrease of $\eta$ would steer the successive average model \eqref{eq:watling} toward WE. With Theorem \ref{thm:convergence-ue}, this question can now be partially answered. 
    Consider the following two DTD models.
    \begin{itemize}
        \item  Model I: $\vs^t = \vs^{t - 1} + c(\vp^{t - 1})$ and $\vp^t = q_r(\vs^t)$.
        \item Model II: $\vs^t = (1 - \eta^t) \cdot \vs^{t - 1} + \eta^t \cdot c(\vp^{t - 1})$ and $\vp^t = q_{r^t}(\vs^t)$ with $\eta^t = 1 / (t +1)$ and $r^t = r \cdot (t + 1)$.
    \end{itemize}
    Model I is a {CumLog} model, and Model II is a successive average model.
    {Mathematically, the two models are identical in the sequence of $\{\vp^t\}_{t = 0}$ they generate (the proof is omitted for brevity).}
    Since Model I is a {CumLog} model with $\eta = 1$, its convergence is guaranteed for sufficiently small $r$ as per Condition (ii) in Theorem \ref{thm:convergence-ue}.  This means the successive average model converges to WE, too, if $\eta^t$ decreases at a rate of $\mathcal{O}(1/t)$ and $r^t$ increases at a rate of $\mathcal{O}(t)$. However, this peculiar coupling between $\eta^t$ and $r^t$ in Model II does not seem to have a reasonable behavioral explanation. In Section \ref{sec:exp-iii}, we shall show very slight modifications of the coupling mechanism could lead to vastly different convergence patterns.
\end{remark}
}

\section{Numerical Examples}
\label{sec:experiment}

{The proposed  CumLog dynamical process is tested on two small networks: a 3-node-4-link (3N4L) network \citep{friesz1990sensitivity} and the Sioux-Falls network \citep{leblanc1975algorithm}.}

{
\textbf{3N4L.} The 3N4L network, as shown in Figure \ref{fig:3n4l}, has 3 nodes, 4 links, and 1 OD pair. It has four routes connecting the origin (node 1) and the destination (node 3). For ease of reference, let us say Route 1 uses links 2 and 4, Route 2 uses links 1 and 4, Route 3 uses links 2 and 3, and route 4 uses links 1 and 3. The number of travelers from node 1 to node 3 is 10. Given the flow $\evx_a$ on link $a$, we model its costs as $u_a = \evh_a + \evw_a \cdot \evx_a^4$, where $[h_1, h_2, h_3, h_4]^{\T} = [4, 20, 1, 30]^{\T}$ and $[\evw_1, \evw_2, \evw_3, \evw_4]^{\T} = [1, 5, 30, 1]^{\T}$.  
\begin{figure}[ht]
    \centering
    \includegraphics[width=0.265\textwidth]{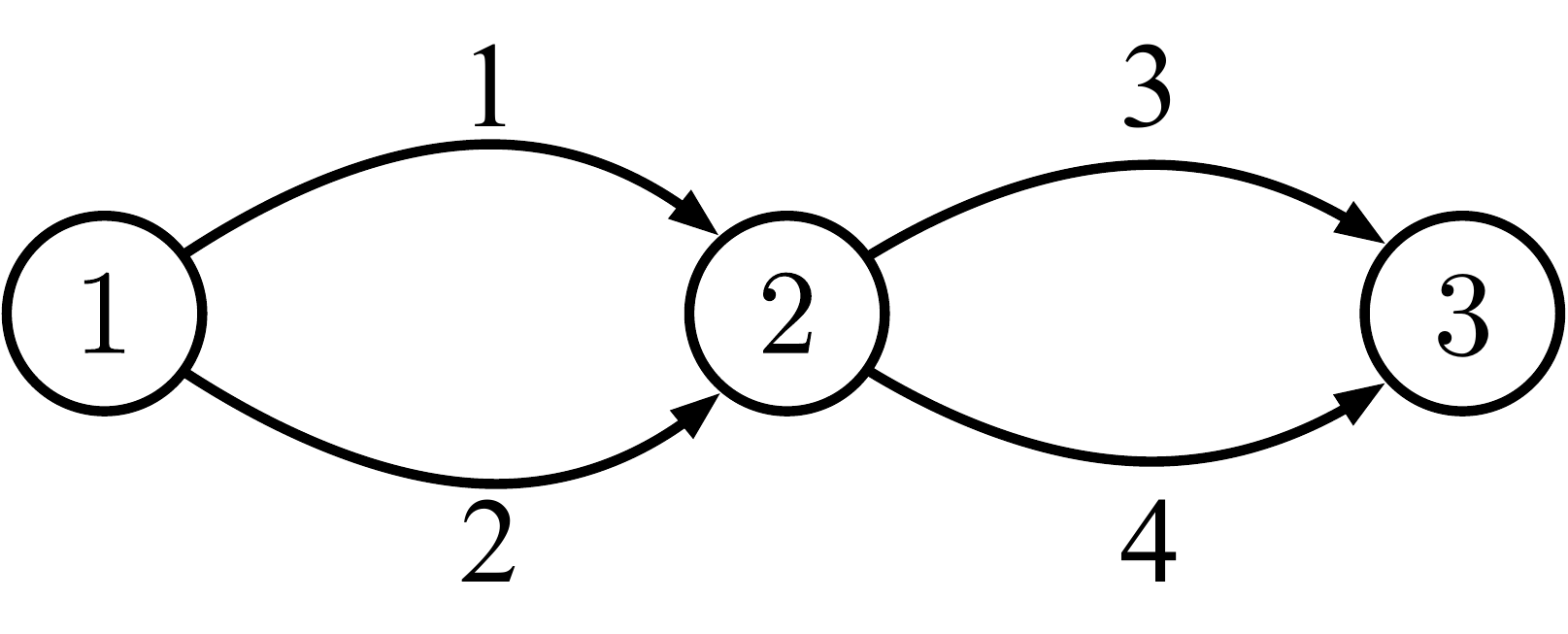}
    \captionof{figure}{The 3N4L network.}
    \label{fig:3n4l}
\end{figure}

\textbf{Sioux-Falls.} The Sioux-Falls network has 24 nodes, 76 links,
and 528 OD pairs. We refer the readers to \citet{leblanc1975algorithm} for the topology, travel demand, and cost function of the Sioux-Falls network.
}

On the two networks, we shall perform four sets of experiments: the first three run on the 3N4L network, {while the fourth on both the 3N4L and the Sioux-Falls network}. 
Section \ref{sec:exp-ii} tests the {CumLog} model under various behavioral parameters. In Section \ref{sec:exp-iii}, we explore how the classical successive average model like those studied by \cite{horowitz1984stability} can be redirected toward a WE, under the guidance of our theory. %
In Section \ref{sec:exp-vii}, we test the {CumLog} model with heterogeneous travelers who exhibit different sensitivity to route valuations. Finally, Section \ref{sec:exp-v} investigates the difference between the WE strategies reached by the {CumLog} model starting from different initial points. 

For a route choice strategy $\vp$, we use the relative gap of its corresponding link flow $\vx \in \sX = \{\vx: \vx = \bar \mLambda \vp, \vp \in \sP\}$, denoted as $\delta(\vx)$, to assess its distance from WE. The relative gap is computed as
\begin{equation}
    \delta(\vx) =-\frac{\langle u(\vx), \vx' - \vx \rangle}{\langle u(\vx), \vx \rangle}, \quad \vx' \in \argmin_{\vx'' \in \sX}~\langle u(\vx), \vx''  \rangle.
\end{equation}

\subsection{Test of convergence conditions}
\label{sec:exp-ii}

We first examine Conditions (i) and (ii) given by Theorem \ref{thm:convergence-ue}. To test Condition (i), we set $\eta^t = 1 / (t + 1)$ for $r = 10, 20, 40$; to test Condition (ii), we fix $\eta^t = 1$ for $r = 0.25, 0.5, 1, 2.5$. For each setting, we run the {CumLog} model starting from $\vs^0 = [0, 0, 0, 0]^{\T}$ until one of the following criteria is met (i) the number of iterations reaches 120; (ii) the relative gap drops below $10^{-9}$, or (iii) the algorithm begins to diverge. 

The convergence patterns reported in Figure \ref{fig:exp-ii} generally agree with the prediction of Theorem \ref{thm:convergence-ue}. Under Condition (i), convergence is ensured regardless of the value of $r$. Interestingly, the larger the value of $r$, the slower the convergence at the beginning and the faster the convergence at the end (see Figure \ref{fig:exp-ii}-(i)). For Condition (ii), the convergence can only be guaranteed when $r$ is sufficiently small. In this case, we can observe from Figure \ref{fig:exp-ii}-(ii) that the convergence rate increases with $r$ when $r\le 1$. Indeed, setting $r=1$ and $\eta=1$ delivers the best performance among all scenarios: it reaches the target with less than 30 iterations. However, when $r = 2.5$, the {CumLog} model failed to converge. 

\begin{figure}[ht]
    \centering
    \begin{subfigure}[b]{0.39\textwidth}
        \includegraphics[width=0.93\columnwidth]{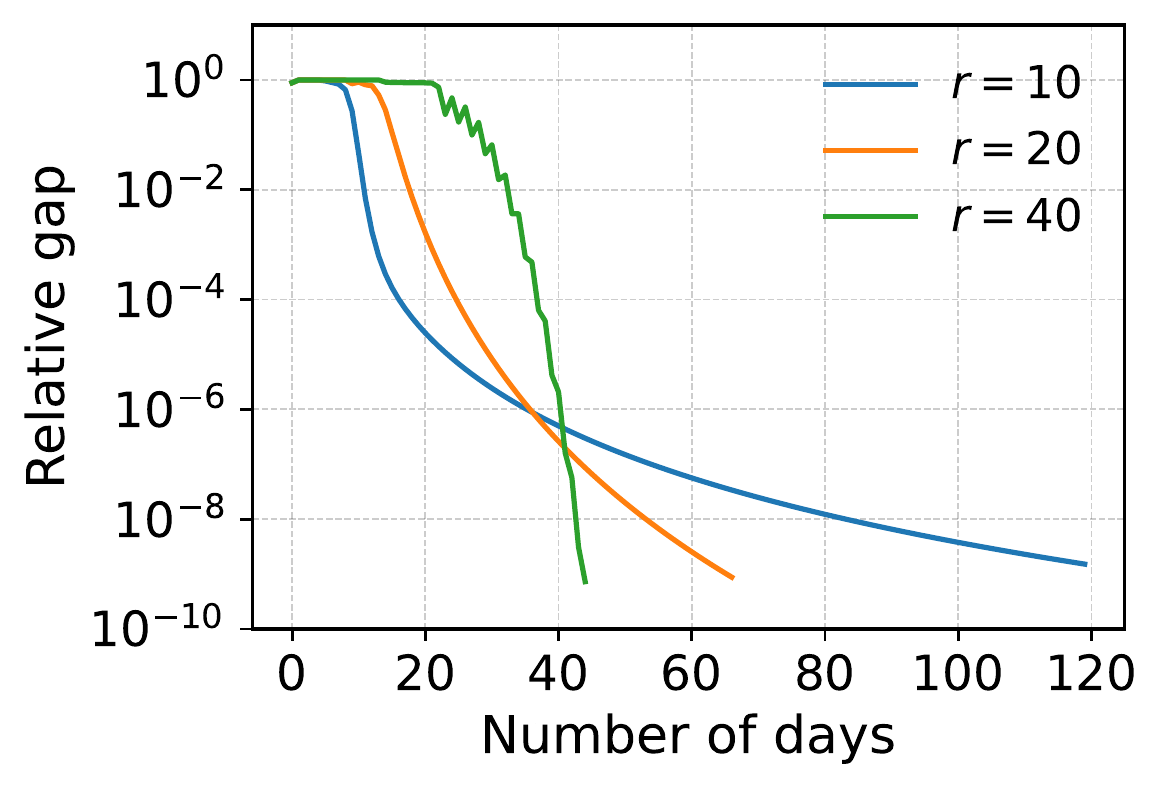}
        \caption{$\eta^t = 1 / (t + 1)$.}
        \label{fig:ii-1}
    \end{subfigure}
    \begin{subfigure}[b]{0.39\textwidth}
        \includegraphics[width=0.93\columnwidth]{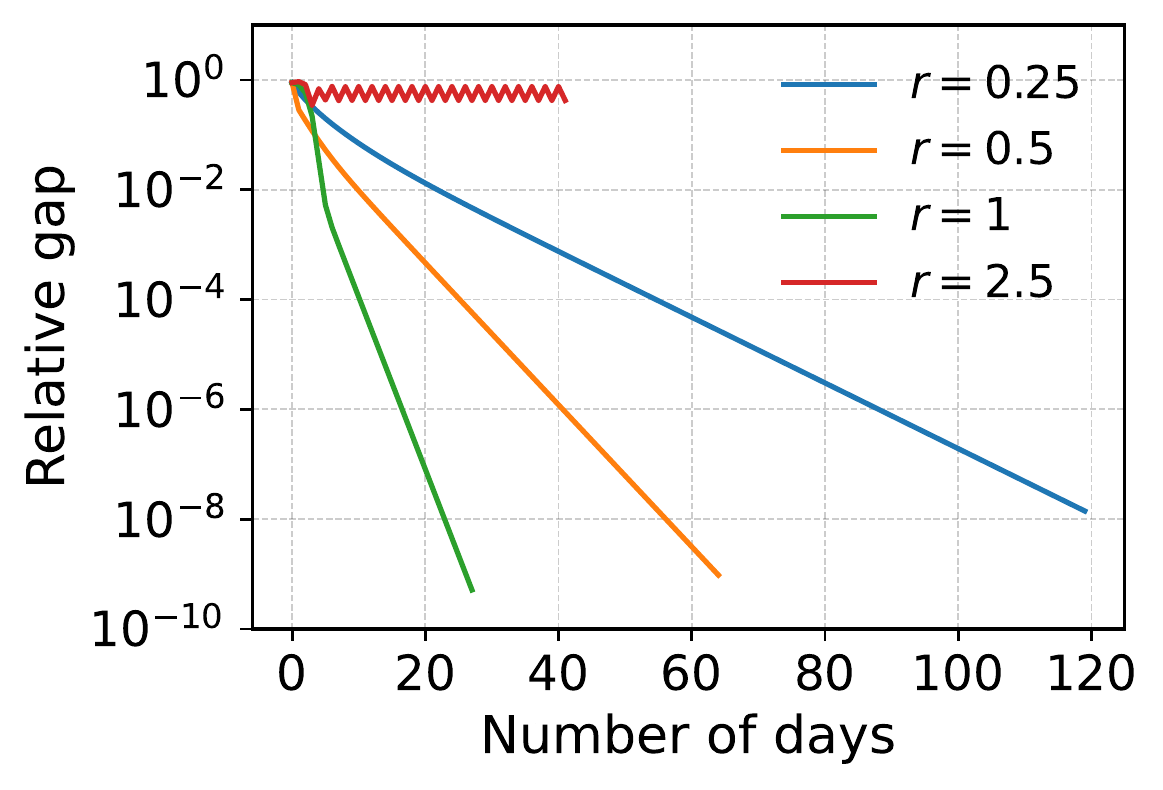}
        \caption{$\eta^t = 1$.}
        \label{fig:ii-2}
    \end{subfigure}
    \caption{Convergence pattern of the {CumLog} model under Conditions (i) and (ii).}
    \label{fig:exp-ii}
\end{figure}

We next fix $r = 1$ and examine how the decreasing rate of $\eta^t$ affects the convergence performance. We set $\eta^t = (t + 1)^{\alpha}$ and report the convergence patterns corresponding to $\alpha = -0.5, -0.25, 0, 0.25$ in Figure \ref{fig:exp-ii-3}. The convergence of the model under $\alpha = -0.5, -0.25, 0$ is guaranteed by Theorem \ref{thm:convergence-ue}.  The faster the $\eta^t$ decreases with $t$, the more quickly travelers tend to settle down, and the slower the convergence. When the decreasing rate is zero (i.e., $\eta^t$ becomes a constant), the convergence is the fastest. However, when the trend is reversed, and $\eta^t$ begins to increase with $t$ ($\alpha = 0.25$), the process quickly diverges (see the red line). This is expected as neither Condition (i) nor Condition (ii) would be satisfied with $r =1$ and $\eta^t = (t + 1)^{0.25}$.

\begin{figure}[ht]
\centering
    \begin{minipage}[t]{0.48\textwidth}
    \centering
    \includegraphics[height=0.54\textwidth]{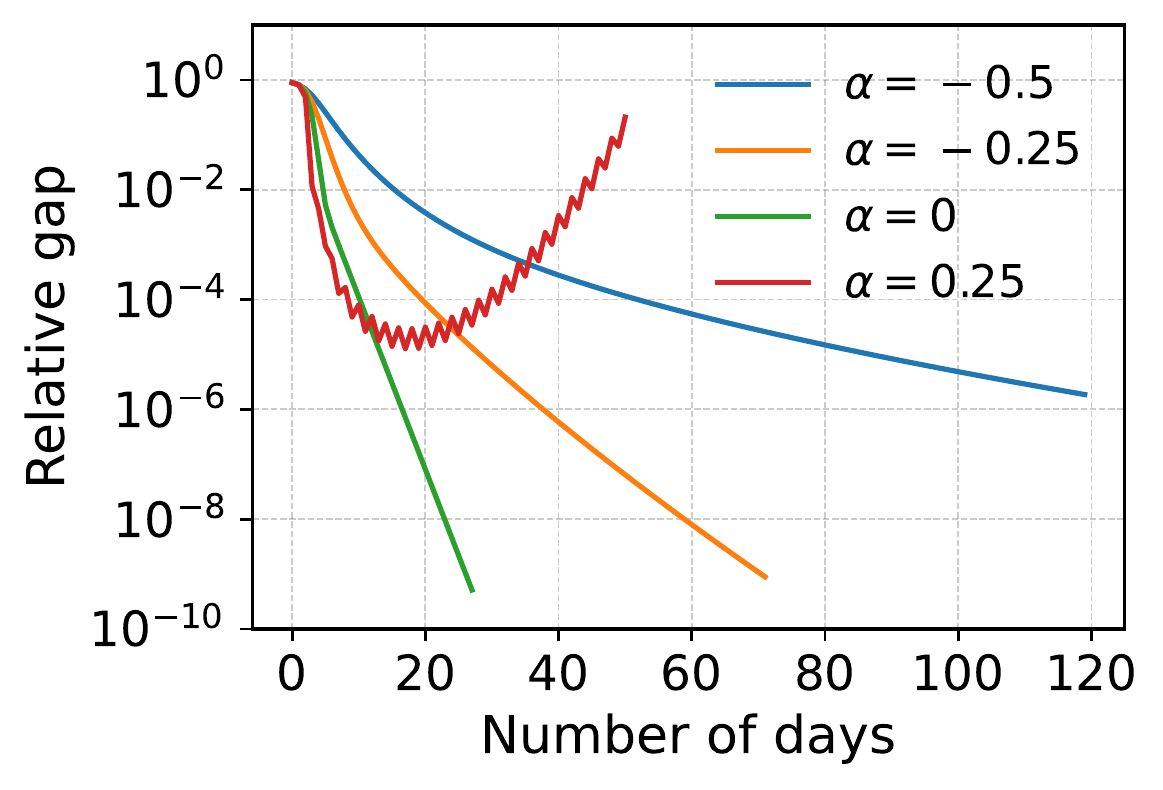}
    \captionof{figure}{Convergence pattern of the {CumLog} model with  $r=1$ and $\eta^t = (1 + t)^{\alpha}$.}
    \label{fig:exp-ii-3}
    \end{minipage}
    \hspace{5pt}
    \begin{minipage}[t]{0.48\textwidth}
    \centering
    \includegraphics[height=0.54\textwidth]{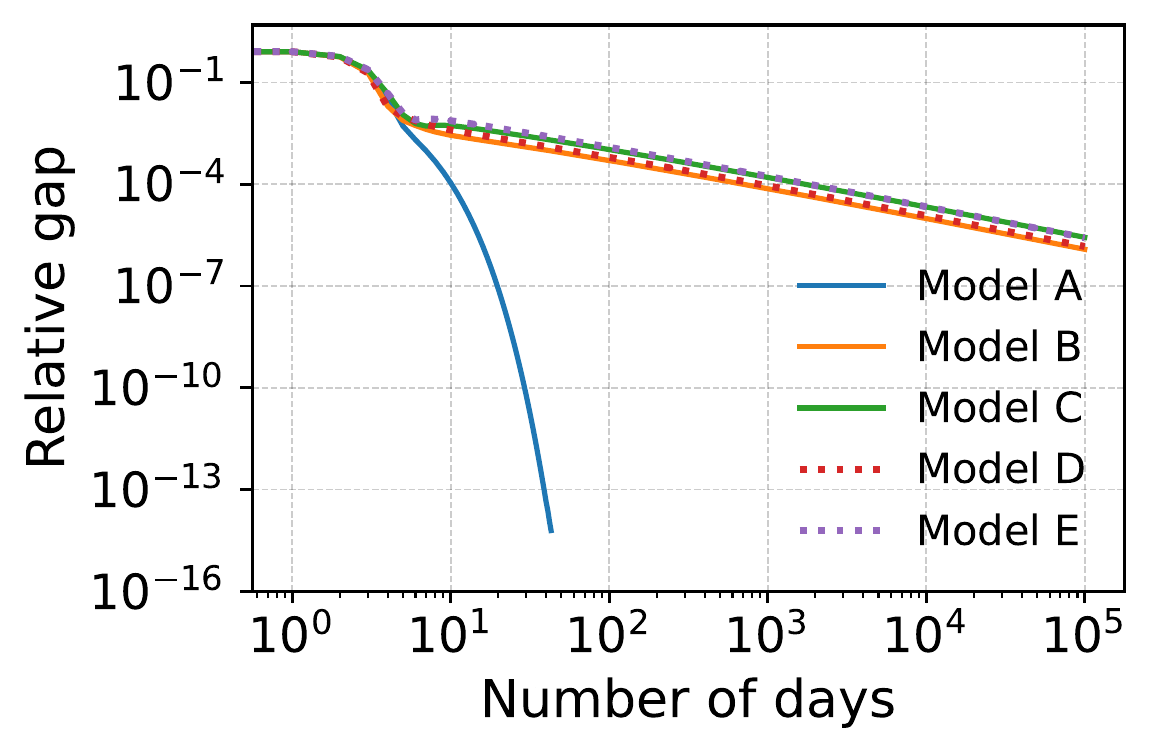}
    \captionof{figure}{Convergence pattern of the SA model with $r^t$  and $\eta^t$ of different changing rates.}
    \label{fig:exp-iii}
    \end{minipage}
\end{figure}

\subsection{Convergence of revised successive average (SA) models to WE}
\label{sec:exp-iii}
  
In the second experiment, we attempt to manipulate the classical successive average (SA)  model using our theory so that it converges to WE rather than SUE. The basic idea is to gradually raise the value of $r$ (the {exploitation} parameter) while reducing  $\eta$ (the {proactivity} measure), as discussed earlier in Remark \ref{rm:eta-r}. We test five models, all based on the SA process $\vs^t = (1 - \eta^t) \cdot \vs^{t - 1} + \eta^t \cdot c(\vp^{t - 1})$ and $\vp^t = q_{r^t}(\vs^t)$. However, $\eta^t$ and $r^t$ are set differently, as detailed below. Model A: $\eta^t = 1 / (t + 1)$ and $r^t = t + 1$. Model B:  $\eta^t = 1 / (t + 1)^{0.99}$ and $r^t = t + 1$. Model C:  $\eta^t = 1 / (t + 1)^{1.01}$ and $r^t = t + 1$. Model D: $\eta^t = 1 / (t + 1)^{0.99}$ and $r^t = (t + 1)^{0.99}$.  Model E:  $\eta^t = 1 / (t + 1)^{1.01}$ and $r^t = (t + 1)^{1.01}$. Model A, as pointed out in Remark \ref{rm:eta-r}, is equivalent to the {CumLog} model with $r = 1$ and $\eta = 1$. Hence, its convergence to WE is guaranteed by Theorem \ref{thm:convergence-ue} and already confirmed in Section \ref{sec:exp-ii}. What we try to examine is the robustness of the ``perturbed" successive average model. Specifically, what happens if we slightly perturb the changing rates of the two parameters? If the manipulated model is robust, then such perturbations should not have a significant impact on the convergence pattern.   Compared to Model A, Models B and C keep the same increasing rate for $r^t$ but slightly modify the decreasing rate of $\eta^t$; Models D and E change both $r^t$ and $\eta^t$ but keep $r^t \cdot \eta^t = 1$ as in Model A. We start all models from $\vs^0 = [0, 0, 0, 0]^{\T}$ and report the convergence pattern in Figure \ref{fig:exp-iii}.

Surprisingly, while Model A converges quickly as expected, none of its four slightly perturbed versions was able to converge at a similar speed --- not even close. In fact, based on the trend, it is unclear whether they would ever converge to a point sufficiently close to WE. Figure \ref{fig:exp-iii} indicates after 100,000 days (or 274 years), they are still far away from reaching the target precision (relative gap of $10^{-9}$). We do not know what caused this dramatic slow-down when the perturbation moves the parameters so slightly away from the trajectory charted by Theorem \ref{thm:convergence-ue}. Indeed, if we compare Model A to Models B and C, the only difference is that $\eta^t$ is changed from $(1+t)^{-1}$ in Model A to $(1+t)^{-0.99}$ in Model B --- which slightly slows down the decreasing rate of $\eta^t$ --- and $(1+t)^{-1.01}$ in Model C --- which slightly speeds up how $\eta^t$ is decreased. Yet, Model A converges within less than a month, at least four to five orders of magnitude faster than both Models B and C.  Regardless of the cause, the phenomenon draws a sharp contrast with the robustness of the {CumLog} model against the changing rate in $\eta^t$. In Figures \ref{fig:exp-ii} and \ref{fig:exp-ii-3}, the {CumLog} model's convergence speed varies within a much narrower range despite much greater variations applied to the parameters.

\subsection{User heterogeneity at WE}
\label{sec:exp-vii}

We next construct an experiment in which travelers differ from each other in terms of route choice behaviors. Specifically, travelers are divided into four classes (Classes 1-4) with the {exploitation} parameter $r$ set to 0.01, 0.1, 1, \text{and~} 10, respectively. As noted before, a larger $r$ suggests a smaller perception error, a greater sensitivity to route evaluations, or a stronger propensity for exploitation, depending on the preferred interpretation of the modeler. The total demand remains the same as in the first two experiments but is equally allocated to the four classes. The travelers from different classes are identical in every aspect except for the value of $r$. Their initial valuation of the routes are $[0, 0, 0, 0]^{\T}$ and their {proactivity} measure  $\eta^t$ is set to a constant of $1$. 
\begin{figure}[ht]
    \centering    \includegraphics[width=0.7\textwidth]{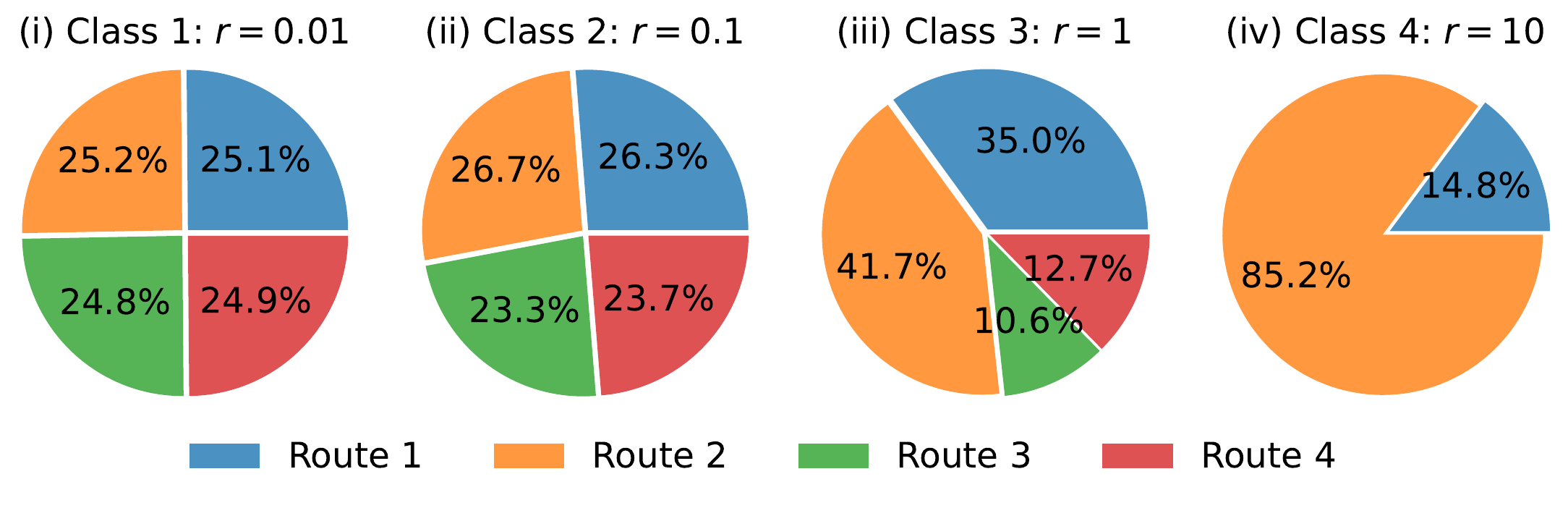}
    \captionof{figure}{WE route choice strategies of travelers with different {exploitation} parameters ($r = 0.01, 0.1, 1, 10$).  }
    \label{fig:exp-vii}
\end{figure}
 We ran the {CumLog} model for 1000 iterations and reached a relative gap below $10^{-14}$. The WE strategy obtained corresponds to exactly the same WE link flow as in the homogeneous case. Figure \ref{fig:exp-vii} reports the WE route choice strategies of each class. We can see all but Class 4 rank the four routes in the same order: Route 2 $>$ Route 1 $>$ Route 4 $>$ Route 3. Class 4 is different only because it does not use the two lower-ranked routes (4 and 3). This indicates that, in this setting, different travelers value the routes the same way but react to the valuations distinctively. 
Moreover, the class with a larger $r$ is more concentrated on the higher-ranked routes. Class 1 was almost indifferent among the four routes, whereas Class 4 completely abandoned Routes 4 and 3.
 
What is remarkable about the above result is that it illustrates WE is compatible with not only bounded rationality but also user heterogeneity. At a WE, some travelers may stick to one or very few routes because they are too rational (or cost-sensitive) to tolerate inferior routes. On the other end of the spectrum are those who are open to exploring all acceptable options with similar probabilities, even the routes with much worse valuations. Still more travelers would fall between the two extremes. The choices of the travelers, as diverse as they are, still result in the same network traffic conditions as if everyone behaves identically and rationally. Therefore, the {CumLog} model allows us to simultaneously accept that WE approximately exist in the real world at the aggregate level and reject the implausible implication that every traveler must be same and perfectly rational. 

\subsection{Non-uniqueness of WE strategies}
\label{sec:exp-v}

{
Our last experiment, performed on both the 3N4L and the Sioux-Falls network, is devised to demonstrate that CumLog may reach WE strategies of different properties when initialized from different points. Note that in both networks, the link cost function $u(\vx)$ is strictly monotone. According to \citet{sheffi1985urban}, under such conditions,  the link flow at WE is unique. Denoting the unique link flow at WE as $\vs^*$, the set of route choice strategies at WE can be represented as a polyhedron
\begin{equation}
    \sP^* = \{\vp^* \in \sP: \bar \mLambda \vp^* = \vx^*\}.
    \label{eq:polyhedron}
\end{equation}
A useful property for differentiating $\vp^* \in \sP^*$ is their entropy, which may be interpreted as its likelihood of realization given the information known to the modeler (e.g., satisfying the WE conditions) \citep{wilson2011entropy}.  Mathematically, entropy may be defined as:
$
    \Phi(\vp^*) = -\langle \diag(\vq) \vp^*, \log(\vp^*)\rangle.
$
Different $\vp^* \in \sP^*$ may use different set of routes. For example, an interior point of the polyhedron  \eqref{eq:polyhedron} uses every route that \emph{may} be used at WE, observing a ``no-route-left-behind" policy  \citep{bar1999route}. A standard WE algorithm, however, usually finds only a subset of all possible WE routes (i.e., it admits a solution on the boundary of the polyhedron  \eqref{eq:polyhedron}).
In what follows, we explore how the choice of the initial strategy $\vp^0$ affects the location and entropy of the equilibrium route choice strategy $\vp^*$ as well as the set of used routes.

\textbf{3N4L.} In the 3N4L network, the set of route choice strategies at WE can be written as
    \begin{equation*}
        \sP^* = \{\vp^*: \vp^* = [0.3 - \lambda, 0.4 - \lambda, 0.3 + \lambda, \lambda]^{\T}, \ \lambda \in [0, 0.3] \}.
    \end{equation*}
    To visualize the difference between these solutions, we plot the relation between the entropy of all $\vp^* \in \sP^*$ and their corresponding value of $\lambda$ in Figure \ref{fig:C}. It can be seen that the entropy of $\vp^*$ first increases and then decreases with $\lambda$; the entropy peaks at  $\lambda = 0.12$, which corresponds to $\bar \vp^* = [0.18, 0.28, 0.42, 0.12]^{\T}$, or the ``maximum-entropy" solution.

    \begin{figure}[ht]
    \centering
    \begin{minipage}[t]{0.4\textwidth}
    \centering
    \includegraphics[height=0.5\textwidth]{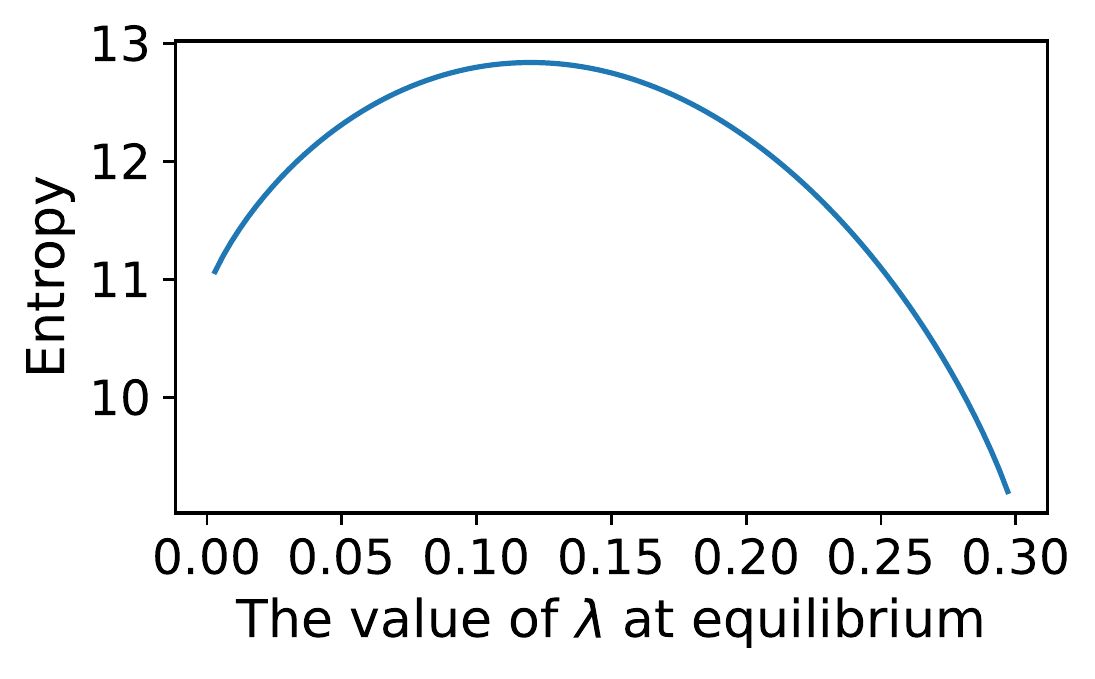}
    \captionof{figure}{Relation between $\Phi(\vp^*)$ for all $\vp^* \in \sP$ and their corresponding value of $\lambda$.}
    \label{fig:C}
    \end{minipage}
    \hspace{8pt}
    \begin{minipage}[t]{0.5\textwidth}
    \centering
    \includegraphics[height=0.4\textwidth]{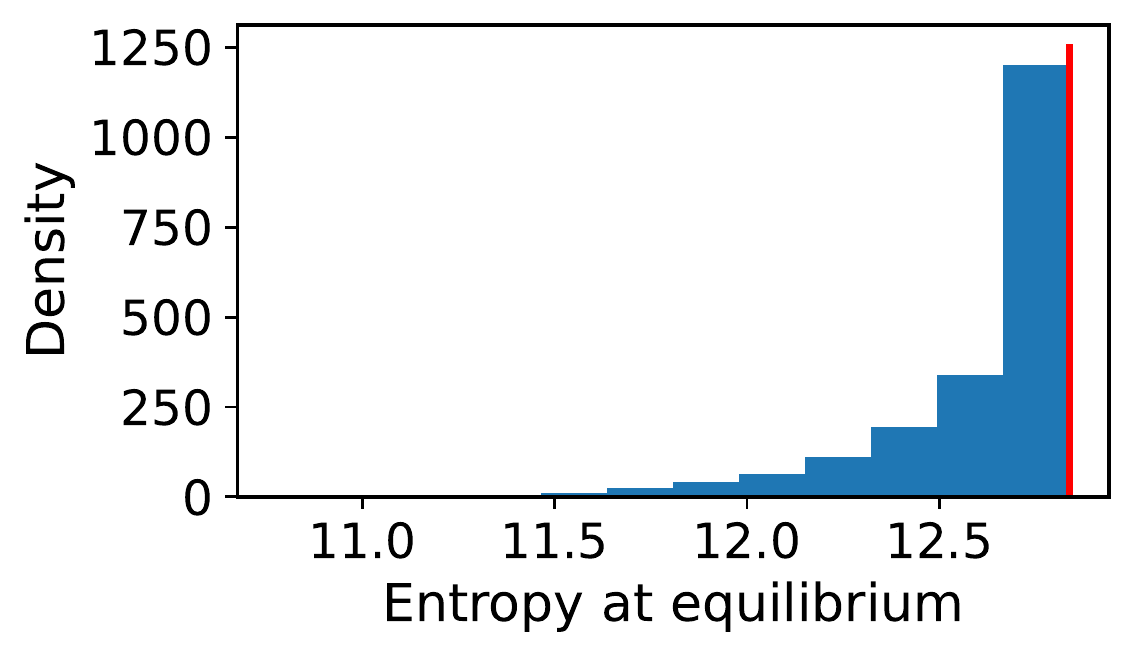}
    \captionof{figure}{Distribution of $\Phi(\vp^*)$ for $\vp^*$ reached by 2000 $\vs^0$. The red line highlights $\Phi(\vp^*)$ for $\vp^*$ reached by $\vs^0 = \vzero$.}
    \label{fig:B}
    \end{minipage}
    \end{figure}

    We then generate a random sample of 2000 $\vs^0$ from a normal distribution and run CumLog starting from an initial point corresponding to each $\vs^0$ in the sample.
    In each run, we set $r = 1$ and $\eta^t = 1$, and terminate it if the relative equilibrium gap is smaller than $10^{-6}$. Figure \ref{fig:B} plots the histogram of the entropy of $\vp^*$ reached by our model. If all the 2000 initial solutions end up at the same WE strategy, then the entropy values would be concentrated at a single point in the histogram. Instead, we find they spread out between a minimum of 10.77 and a maximum of 12.84. Moreover, the equilibrium solution with the highest entropy among the 2000 points corresponds to the initial solution $\vs^0 = \vzero$ (i.e., all valuations are initially set to zero, representing zero information on all routes). A closer look reveals that this solution is indeed the maximum entropy solution, i.e., $\bar \vp^* = [0.18, 0.28, 0.42, 0.12]^{\T}$.
}

{\textbf{Sioux-Falls.} We then investigate how $\vs^0$ affects the set of routes used at the WE reached by CumLog on Sioux-Falls. Using the methods by \citet{xie2019new} and \citet{tobin1988sensitivity}, we find the maximum and minimum numbers of routes to be used at WE are 770 and 557, respectively. This suggests a  WE algorithm may locate a solution whose number of used routes is anywhere between 557 and 770. }
To perform the test, we use a set of 1238 routes that contains all 770 routes that may be used by a WE strategy. 
The initial valuation $\vs^0$ is randomly sampled from a normal distribution, and the sample size is set to 2000. The all-zero initialization $\vs^0 = \vzero$ is employed as a benchmark. In all runs, we set $r = 2.5$, $\eta^t = 1$, and the maximum number of days to 1000. 

\begin{figure}[ht]
    \centering
    \begin{subfigure}[b]{0.42\textwidth}
        \includegraphics[width=0.8\columnwidth]{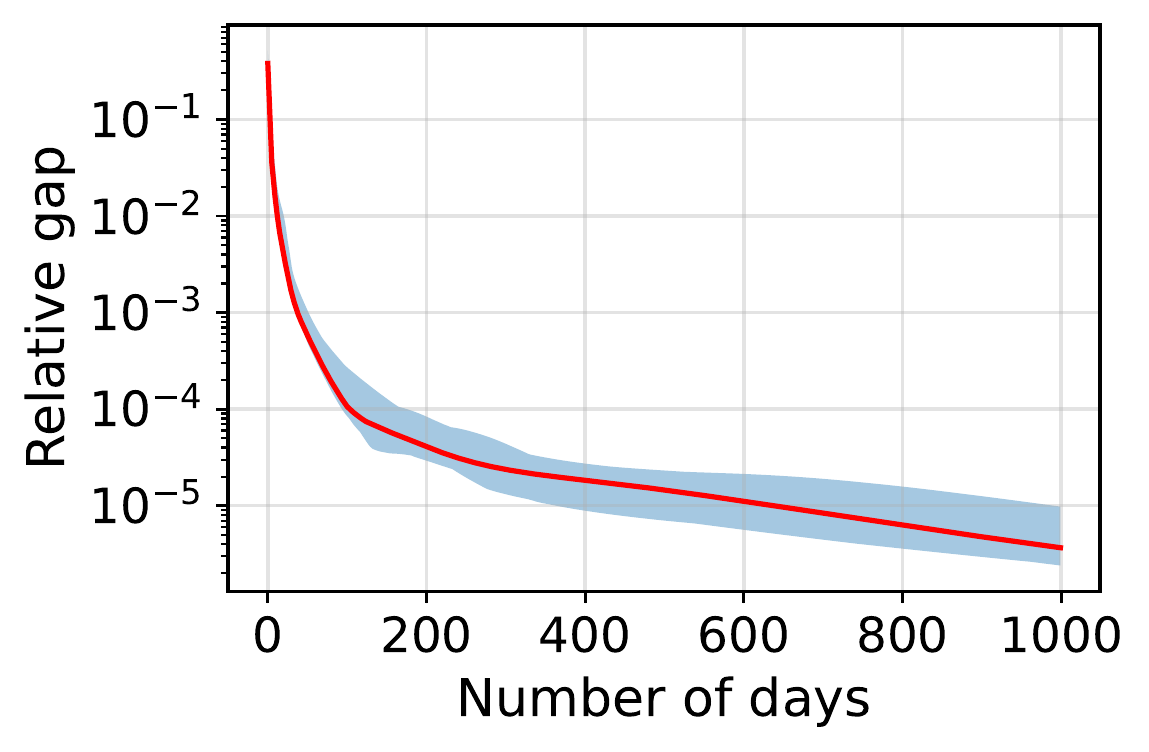}
        \caption{}
        \label{fig:vi-1}
    \end{subfigure}
    \begin{subfigure}[b]{0.42\textwidth}            
        \includegraphics[width=0.8\columnwidth]{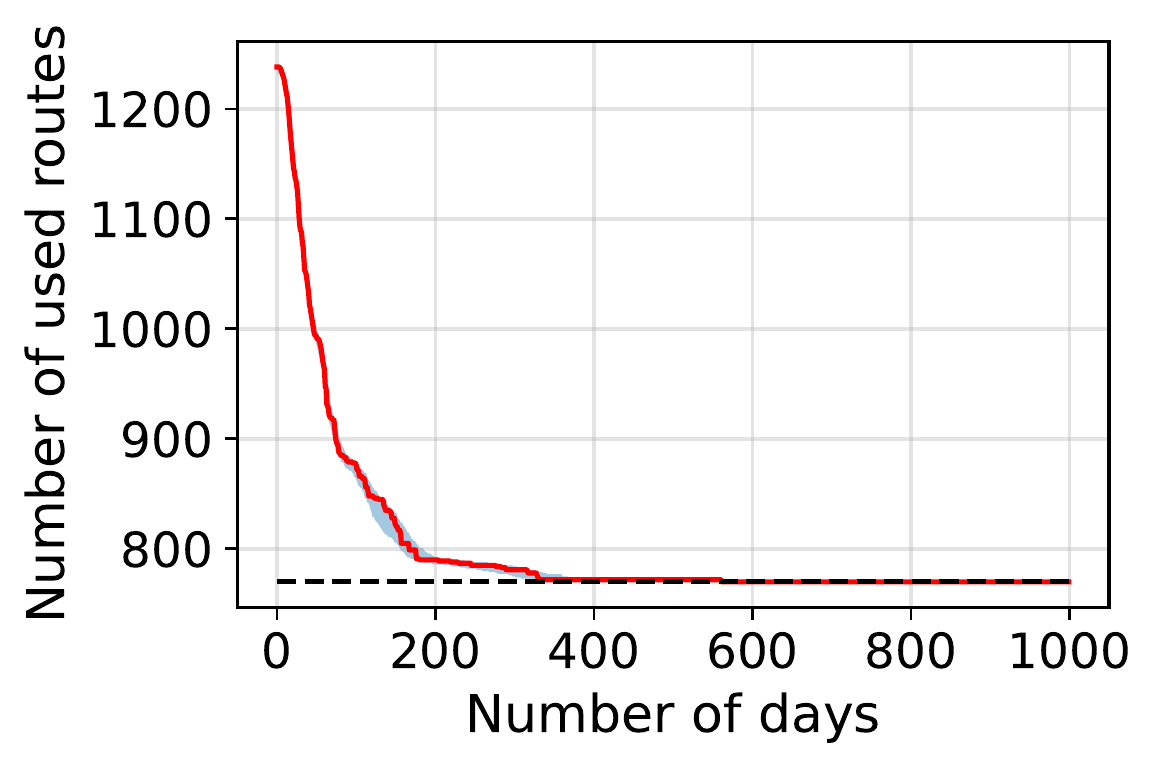}
        \caption{}
    \label{fig:vi-2}
    \end{subfigure}
    \caption{Convergence pattern of {CumLog} with different $\vs^0$. The blue shadow represents the collection of the convergence curves corresponding to 2000 random $\vs^0$; the red solid line represents the convergence curve corresponding to $\vs^0 = \vzero$. The black dashed line in (ii) represents the maximum number of routes that may be used by a WE strategy.}
    \label{fig:exp-vi}
\end{figure}

Figure \ref{fig:exp-vi} shows how the relative gap and the total number of routes actively used by travelers (a route is actively used if its probability of being selected is no less than $10^{-6}$) change with the number of days. We can see the convergence pattern is affected by the initial point, but the impact on the convergence rate is insignificant. The number of actively used paths also descends quickly to the lower bound. Interestingly, the CumLog model never ``accidentally" eliminates a potential WE route, nor does it ever fail to exclude routes that are not supposed to be there --- the number of used routes at the solution reached by CumLog is always 770. If this property can be established analytically, it will make the {CumLog} model a suitable algorithm for finding all WE routes.

\section{Conclusions}
\label{sec:conclusion}

As one of the most fundamental concepts in transportation science,  Wardrop equilibrium (WE)  was the cornerstone of countless large mathematical models that were built in the past six decades to plan, design, and operate transportation systems around the world. However, like Nash Equilibrium, its more famous cousin, WE has always had a somewhat flimsy behavioral foundation. The efforts to strengthen this foundation have largely centered on reckoning with the imperfections in human decision-making processes, such as the lack of accurate information, limited computing power, and sub-optimal choices.  This retreat from behavioral perfectionism was typically accompanied by a conceptual expansion of equilibrium. In place of WE, for example, transportation researchers had defined such generalized equilibrium concepts as stochastic user equilibrium (SUE) and boundedly rational user equilibrium  (BRUE). Invaluable as these alternatives are to enrich our understanding of equilibrium and to advance modeling and computational tools, they advocate for the abandonment of WE, predicated on its incompatibility with real behaviors.  Our study aims to demonstrate that giving up perfect rationality need not force a departure from WE.  To this end, we construct a day-to-day (DTD) dynamical model that mimics how travelers gradually adjust their valuations of routes, hence the choice probabilities, based on past experiences.  

Our model, called cumulative logit ({CumLog}), resembles the classical DTD models but makes a crucial change: whereas the classical models assume routes are valued based on the cost averaged over historical data, ours values the routes based on the cost accumulated. To describe route choice behaviors, the {CumLog} model only uses two parameters, one accounting for the rate at which the future route cost is discounted in the valuation relative to the past ones (the {proactivity} measure) and the other describing the sensitivity of route choice probabilities to valuation differences (the {exploitation} parameter).  
We prove that {CumLog} always converges to WE, regardless of the initial point, as long as the {proactivity} measure either shrinks to zero at a sufficiently slow pace as time proceeds or is held at a sufficiently small constant value.

By equipping WE with a route choice theory compatible with bounded rationality,  we uphold its role as a benchmark in transportation systems analysis.  Compared to the incumbents, our theory requires no modifications of WE as a result of behavioral accommodation.  This simplicity helps avoid the complications that come with a ``moving benchmark," be it caused by a multitude of equilibria or the dependence of equilibrium on certain behavioral traits.  Moreover, by offering a plausible explanation for travelers' preferences among equal-cost routes at WE,  the theory resolves the theoretical challenge posed by Harsanyi's instability problem.  Note that we lay no claim on the behavioral truth about route choices.  Real-world routing games take place in such complicated and ever-evolving environments that they may never reach a true stationary state, much less the prediction of a mathematical model riddled with a myriad of assumptions.  Indeed, a relatively stable traffic pattern in a transportation network may be explained as a point in a BRUE set, an SUE tied to properly calibrated behavioral parameters, or simply a crude WE reached by {CumLog}.  {While more empirical research is needed to vet our theory and compare it with existing ones, we should no longer write off WE simply because it adheres to behavioral perfectionism. }

Other than satisfying theoretical interests, the {CumLog} model may also be used as a prototype algorithm for solving routing games. On large networks, the convergence of the {CumLog} model may be relatively slow (see, for example, Figure \ref{fig:exp-vi}-(i)). This is hardly surprising given no higher-order information (e.g., the derivative of route cost) is employed.  However, if the goal is to find a good approximate solution quickly, then a {CumLog}-based algorithm can be quite competitive thanks to its simplicity (only route costs are needed), flexibility (easy extension to more general settings), and stability (relatively weak requirements for convergence). To be sure, the current {CumLog} model is still far away from a practical algorithm for WE routing games. Of the missing components, the most important is an efficient route-generation scheme. We leave the development of such an algorithm to future investigations.

Numerical experiments in Section \ref{sec:exp-v} revealed a few noteworthy phenomena. First, the {CumLog} model is capable of identifying all routes that may be used by any WE strategy.  Second, the WE strategy resulting from the dynamical process is closely related to the initial route valuation. In particular, it seems that an all-zero initial valuation leads to the entropy-maximizing (or most likely) WE strategy.  Does this mean the {CumLog} model can be used to guide the selection of a unique WE strategy, especially in locating the most likely one? We leave this question also to a future study. 

\begin{appendices}
\numberwithin{algorithm}{section}
\numberwithin{equation}{section}
\numberwithin{figure}{section}
\numberwithin{table}{section}

\section{Proofs of Main Results}

We start with necessary definitions and results related to vector norm, variational inequality problem (VIP), Fenchel conjugate, negative entropy function,  KL divergence, and logit model.

\smallskip
\textbf{Vector norm.} 

\begin{lemma}[\citet{boyd2004convex}]
\label{lm:l1-l2}
    For all $\va \in \sR^n$, we have $\|\va\|_2 \leq \|\va\|_1$.
\end{lemma}

\begin{lemma}[\citet{marcotte1995convergence}]
\label{lm:bound}
    For any two $\va, \va' \in \sR^n$, we have $\langle\va, \va'\rangle \leq \|\va'\|_2^2 + \frac{1}{4} \cdot \|\va\|_2^2$.
\end{lemma}

\smallskip
\textbf{Variational inequality problem (VIP).} 

The following result characterizes the relationship between a convex program and a VIP.  
\begin{proposition}[\citet{kinderlehrer2000introduction}] 
\label{eg:optimization}
    Suppose that $\sA \subseteq \sR^n$ is a convex set and $f: \sA \to \sR$ is a continuously differentiable and convex function, then $\va^* \in \sA$ minimizes $f(\va)$ if and only if
    $
        \langle\nabla f(\va^*), \va - \va^*\rangle \geq 0
    $
    for all $\va \in \sA$.
\end{proposition}

\smallskip
\textbf{Fenchel conjugate and Fenchel dual.}

The Fenchel conjugate (also known as convex conjugate or Fenchel transformation) of a scalar function is defined as follows.
\begin{definition}[Fenchel conjugate]
    Given a function $f: \sA \subseteq \sR^{n}  \to \bar \sR$, its Fenchel conjugate $f^*: \bar \sR^{n} \to \sR$ is defined as
    $f^*(\vb) = \sup_{\va \in \sA} \{\langle \va, \vb\rangle - f(\va)\}$.
\end{definition}

If $f^*(\vb)$ is the Fenchel conjugate of $f(\va)$, then the gradient of $f^*(\vb)$ and $f(\va)$ --- provided that they are both differentiable --- has a dual relationship as described in the following lemma.
\begin{lemma}[\citet{hiriart2004fundamentals}]
\label{lm:fenchel}
    If $\sA \subseteq \sR^{n}$ is a closed set and $f: \sA \to \bar \sR$ is a differentiable and convex function, then for any $\va \in \sR^n$ and $\vb \in \bar \sR^n$, the following two conditions are equivalent: (1) $\vb = \nabla f(\va)$ and (2) $\va = \nabla f^*(\vb)$. We call $a$ and $b$ each other's Fenchel dual. 
\end{lemma}

\smallskip
\textbf{Negative entropy and Kullback–Leibler divergence.}

The following lemma describes the Fenchel conjugate of the negative entropy function $\phi_w(\vp_w) = \langle \vp_w, \log(\vp_w) \rangle$.
\begin{lemma}[\citet{boyd2004convex}]
\label{eg:entropy-dual}
    The Fenchel conjugate of $\phi_w(\vp_w)$ is $\phi_w^*(\vq_w) = \langle \vone, \exp(\vq_w - 1)\rangle$, which implies $\nabla \phi_w(\vp_w) = \log(\vp_w) + 1$ and $\nabla \phi_w^*(\vq_w) = \exp(\vq_w - 1)$.
\end{lemma}

The Kullback–Leibler (KL) divergence is widely used to measure the ``statistical distance" between probability distributions, defined as follows. 

\begin{definition}[KL divergence]
\label{def:kl}
    For two $\vp_w, \vp_w'$ in $\sP_w$, the KL divergence between $\vp_w$ and $\vp_w'$ is defined as 
    $$D_w(\vp_w, \vp_w') = \phi_w(\vp_w) - \phi_w(\vp_w') - \langle \nabla \phi_w(\vp_w'), \vp_w - \vp_w' \rangle = \langle \vp_w, \log \vp_w - \log \vp_w' \rangle.$$
\end{definition}
\begin{lemma}
\label{lm:kl-finite}
    For any $\vp_w, \vp'_w \in \sP_w$,  $D_{\phi_w}(\vp_w, \vp_w') < \infty$ if and only if $\supp(\vp_w) \subseteq \supp(\vp_w')$.
\end{lemma}

The following lemma is often known as the generalized triangle inequality for KL divergence.
\begin{lemma}[\citet{chen1993convergence}]
\label{lm:triangle}
    For any three points $\vp_w, \vp_w', \vp_w'' \in \sP_w$, we have
    \begin{equation}
        D_w(\vp_w, \vp_w') + D_w(\vp_w', \vp_w'') - D_w(\vp_w, \vp_w'') = \langle \vp_w - \vp_w', \nabla \phi_w(\vp_w'') - \nabla \phi_w(\vp_w') \rangle.
	\label{eq:triangle}
    \end{equation}
\end{lemma}

The KL divergence can be lower-bounded by the $\ell_1$ norm.
\begin{lemma}[\citet{beck2003mirror}]
\label{lm:bregman-ieq}
    For any $\vp_w, \vp_w' \in \sP_w$, we have $D_w(\vp_w, \vp_w') \geq 1/2 \cdot  \|\vp_w - \vp_w'\|_1^2$.
\end{lemma}

\begin{definition}[KL projection]\label{def:KLproj}
    For any $\vp_w \in \sR^{|\sK_w|}$, the KL projection of $\vp_w$ on $\sP_w$ is defined as\footnote{The KL projection is unique because the KL divergence $D_w(\vp_w', \vp_w)$ is strongly convex in the first argument $\vp_w'$.}
    \begin{equation}
        g_w(\vp_w)  = \argmin_{\vp_w' \in \sP_w} \, D_w(\vp_w', \vp_w).
    \end{equation}
\end{definition}

The next result asserts that the KL projection is analytic when $\sP_w$ is a probability simplex.
\begin{lemma}[\citet{beck2003mirror}]
\label{lm:kl-projection}
    For any $\vp_w \in \sR^{|\sP_w|}$, we have
    $g_w(\vp_w) = \vp_w / \|\vp_w\|_1$.
\end{lemma}

\smallskip
\textbf{Logit model.}

\begin{lemma}
[\citet{ben1985discrete}]
\label{lm:difference-logit}
    For any $r > 0$, given two vectors $\vs_w, \vs_w' \in \bar \sR^{|\sK_w|}$, 
    \begin{equation}
        \frac{\exp(-r \cdot \vs_w)}{\|\exp(-r \cdot \vs_w\|_1} = \frac{\exp(-r \cdot \vs_w')}{\|\exp(-r \cdot \vs_w'\|_1}
    \end{equation}
if and only if there exists $\eva_w \in \sR$ such that $\vs_w' = \vs_w + \eva_w \cdot \vone_w$.
\end{lemma}
Lemma \ref{lm:difference-logit} indicates that logit choices only depend on the difference in disutilities. 

\begin{proposition}
\label{lm:projection-logit}
    Given any $r > 0$, $\vs \in \bar \sR^{|\sK|}$, and $\vp = q_r(\vs)$, we have $\vp_w = g_w(\nabla \phi_w^*(-r \cdot \vs_w))$ for all $w \in \sW$, where $g_w(\cdot)$ is KL projection (see Definition \ref{def:KLproj}). 
\end{proposition}
\begin{proof}
    For all $w \in \sW$, we have
    \begin{equation}
        g_w(\nabla \phi_w^*(-r \cdot \vs_w)) = g_w(\exp(-r \cdot \vs_w - 1)) = \frac{\exp(- r \cdot \vs_w - \vone_w)}{\|\exp(- r \cdot \vs_w - \vone_w)\|_1} = \frac{\exp(- r \cdot \vs_w)}{\|\exp(- r \cdot \vs_w)\|_1} = \vp_w,
    \end{equation}
    where the three equalities follow from, in the same order, from Lemmas \ref{eg:entropy-dual}, \ref{lm:kl-projection}, and \ref{lm:difference-logit}.
\end{proof}

Proposition \ref{lm:projection-logit} states that given any disutility vector $\vs \in \bar \sR^{|\sK|}$, the corresponding logit choice $\vp = q_r(\vs)$ can be obtained by KL-projecting the Fenchel dual of $-r \cdot \vs$ on $\sP$. This property is widely used in optimization, convex analysis, and information theory. 

\subsection{Proof of Proposition \ref{prop:property}}
\label{app:property}

The following lemma gives sufficient conditions for checking cocoercivity. 
\begin{lemma}[\citet{marcotte1995convergence}]
\label{lm:cocoercive}
    Suppose that a function $f: \sA \subseteq \sR^n \to \sR^n$ is  twice continuously differentiable and $L$-Lipschitz continuous. If both $\nabla f(\va)$ and $(\nabla f(\va))^2$ are positive semi-definite on $\sA$ for all $\va \in \sA$ , then $f$ is $\mu$-cocoercive ($\mu = 1/4L$) on $\sA$.
\end{lemma}

We are now ready to prove Proposition \ref{prop:property}.

\begin{proof}[Proof of Proposition \ref{prop:property}]

    \text{Property (i)}.  As $\nabla c(\vp) = \bar \mLambda^{\T} \nabla u(\bar \mLambda \vp) \mLambda$, where $\bar \mLambda = \mLambda \diag(\mSigma^{\T} \vd)$, the symmetry (resp., positive semi-definiteness) of $\nabla u(\vx)$ implies of the symmetry (resp., positive semi-definiteness) of $c(\vp)$. Hence,  by Lemma \ref{lm:cocoercive}, we  claim $c(\vp)$ is $1/4L$-cocoercive  ($L = \max_{\vp \in \sP} \|\nabla c(\vp)\|_2$) under Assumptions \ref{ass:differentiable} and \ref{ass:cocoercive} whenever $\nabla u(\vx)$ is symmetric. If $\nabla u(\vx)$ is asymmetric, Lemma \ref{lm:cocoercive} still implies the $1/4H$-cocoercivity of $u(\vx)$ on $\sX = \{\vx \in \sR_+^{|\sE|}: \vx = \bar \mLambda \vp, \ \vp \in \sP\}$, where $H = \max_{\vx \in \sX} \|u(\vx)\|_2$. Thus, given any two $\vp, \vp' \in \sP$ and setting $\vx = \bar \mLambda \vp$ and $\vx' = \bar \mLambda \vp'$, we have
    \begin{equation}
    \begin{split}
        \langle c(\vp) - c(\vp'), \vp - \vp' \rangle &\geq \frac{1}{d_{\max}} \cdot
        \langle u(\vx) - u(\vx'), \vx - \vx' \rangle \\
        &\geq \frac{1}{4d_{\max} \cdot H} \cdot \|u(\vx) - u(\vx')\|_2^2 \geq \frac{1}{4d_{\max} \cdot H \cdot \|\mLambda\|_2^2} \cdot \|c(\vp) - c(\vp')\|_2^2,
    \end{split}
    \end{equation}
    where $d_{\max} = \max_{w \in \sW} \{\evd_w\}$ gives the first inequality, the cocoercivity of $u(\vx)$ leads to the second inequality, and the third follows from classical results in linear algebra. %
    Hence,  $c(\vp)$ is  $1/4L$-cocoercive if we set $L = d_{\max} \cdot H \cdot \|\mLambda\|_2^2$.
    
    \text{Property (ii).} Given any $\vp \in \sP$ and $\vp^* \in \sP^*$, we decompose $\langle c(\vp), \vp - \vp^*\rangle$ to
    \begin{equation}
        \langle c(\vp), \vp - \vp^*\rangle = \langle c(\vp) - c(\vp^*), \vp - \vp^*\rangle + \langle c(\vp^*), \vp - \vp^*\rangle,
      \label{eq:decomposition}
    \end{equation}
    where the non-negativity of the first and second terms is guaranteed by, respectively, the monotonicity of $c(\vp)$ and Proposition \ref{prop:ue-vi}. Thus, we have $\langle c(\vp), \vp - \vp^*\rangle \geq 0$. 
  
    \textit{Sufficiency}. If $\vp \in \sP^*$, then Proposition \ref{prop:ue-vi} implies $\langle c(\vp), \vp^* - \vp \rangle \geq 0$, combined with  $\langle c(\vp), \vp - \vp^*\rangle \geq 0$ then leads to $\langle c(\vp), \vp - \vp^*\rangle = 0$. 
    
    \textit{Necessity}. If $\langle c(\vp), \vp - \vp^*\rangle = 0$, then both terms on the right-hand side of Equation \eqref{eq:decomposition} must be zero as they are both non-negative. Property (i) then gives
    $
        0 = \langle c(\vp) - c(\vp^*), \vp - \vp^*\rangle \geq  1/4L \cdot \|c(\vp) - c(\vp^*) \|_2^2,
    $
    which  implies $c(\vp) = c(\vp^*)$.
    Denoting $\vc^* = c(\vp^*)$ and $b_w^* = \min_{k' \in \sK_w} \evc_k^*$ for all $w \in \sW$, we next claim $\evc_k^* > b_w^* \Rightarrow \evp_k = 0$  for all $w \in \sW$ and $k \in \sK_w$.  To see this, suppose there exist $w' \in \sW$ and $k' \in \sK_{w'}$ such that $c^*_{k'} = b_{w'}^* + \delta$ ($\delta > 0$) and $\evp_{k'} > 0$, then we have $\langle \vc^*_{w'}, \vp_{w'} \rangle \geq  b_{w'} + \delta \cdot \evp_{k'} \rightarrow \langle \vc^*, \vp \rangle  = \sum_{w \in \sW} \langle \vc_w^*, \vp_w \rangle \geq \sum_{w \in \sW} b_w^* + \delta \cdot \evp_{k'}$. Noting that $\langle \vc^*, \vp^*\rangle = \sum_{w \in \sW} b_w^*$, we arrive at $\langle \vc^*, \vp \rangle > \langle \vc^*, \vp^*\rangle$, which contradicts $\langle \vc^*, \vp - \vp^*\rangle  = 0$.  Finally, since $c(\vp) = c(\vp^*)$, it follows  $c_k(\vp) > b_w \Rightarrow \evp_k = 0$  for all $w \in \sW$ and $k \in \sK_w$, where $b_w = \min_{k' \in \sK_w} c_k(\vp)
    $.  This confirms that $\vp \in \sP^*$.
    
    \text{Property (iii)}. Given any $\vp^* \in \sP^*$, Property (ii) implies $\sP^* = \{\vp \in \sP: \langle c(\vp), \vp - \vp^*\rangle = 0\}$, which in turn implies $\sP^*$ is closed and hence compact. Given any two $\vp^*, (\vp^*)' \in \sP^*$, we next prove $\bar \vp = \lambda \cdot \vp^* + (1 - \lambda) \cdot (\vp^*)' \in \sP^*$ fro any $\lambda \in [0, 1]$. By Property (ii), we have
    \begin{equation}
        \begin{cases}
            \displaystyle \langle c(\bar \vp), \vp^* - (\vp^*)'\rangle = \frac{1}{\lambda} \cdot \langle c(\bar \vp), \bar \vp - (\vp^*)'\rangle \geq 0, \\[5pt]
            \displaystyle \langle c(\bar \vp), \vp^* - (\vp^*)'\rangle = -\frac{1}{1 - \lambda} \cdot \langle c(\bar \vp), \bar \vp - \vp^*\rangle \leq 0.
        \end{cases}
    \end{equation}
    Combing the two relations yields $\langle c(\bar \vp), \vp^* - (\vp^*)'\rangle = 0\rightarrow \langle c(\bar \vp), \bar \vp - (\vp^*)'\rangle = 0$. Thus, $\sP^*$ is also convex.
\end{proof}

\subsection{Proof of Lemma \ref{lm:iteration}}
\label{app:iteration}
\begin{proof}[Proof of Lemma \ref{lm:iteration}]
    Denoting $\vs^{t + 1} =  \vs^t + \eta^t \cdot c(\vp^t)$, Proposition \ref{lm:projection-logit} implies $\tilde \vp = q_r(\vs^{t + 1})$ if and only if
    \begin{equation}
    \begin{split}
        \tilde \vp_w &= \argmin_{\vp_w \in \sP_w} \, D_w(\vp_w, \nabla \phi_w^*(-r \cdot \vs_w^{t + 1})) = \argmin_{\vp_w \in \sQ_w^t} \, D_w(\vp_w, \nabla \phi_w^*(-r \cdot \vs_w^{t + 1})), \quad \forall w \in \sW,
        \label{eq:fi-1}
    \end{split}
    \end{equation} 
    where the second equality holds because $D_w(\vp_w, \nabla \phi_w^*(-r \cdot \vs_w^{t + 1}))$ is finite if and only if $\vp_w \in \sQ_w^t$ (see Lemma \ref{lm:kl-finite}).  Then Proposition \ref{eg:optimization} implies that $\tilde \vp_w$ solves the KL projection problem \eqref{eq:fi-1} if and only if
    \begin{equation}
        \langle \log(\tilde \vp_w) + 1 - \log(\nabla \phi_w^*(-r \cdot  \vs_w^{t + 1})),  \vp_w - \tilde \vp_w\rangle \geq 0, \quad \forall \vp_w \in \sQ_w^t. 
        \label{eq:i-1}
    \end{equation}
    Consider the term $\log(\vp_w^{t + 1}) + 1 - \log(\nabla \phi_w^*(-r \cdot  \vs_w^{t + 1}))$ in the above inequality. First, we have $\nabla \phi_w(\tilde \vp_w) = \log(\tilde \vp_w) + 1$ by Lemma \ref{eg:entropy-dual}. Second, noting that  there must exist some $\eva_w \in \sR^{|\sW|}$ such that $-r \cdot \vs_w^t = \log(\vp_w^t) + 1$ (see  Lemma \ref{lm:difference-logit}), we
    have
    \begin{equation}
    \begin{split}
        &\log(\nabla \phi_w^*(-r \cdot \vs_w^{t + 1}))  = \nabla \phi_w (\nabla \phi_w^*(-r \cdot \vs_w^{t + 1})) - 1 = -r \cdot \vs_w^{t + 1} - 1 \\
        &\qquad = -\tilde \eta^t \cdot c_w(\vp^t) + \log(\vp_w^t) + \eva_w \cdot \vone_w - 1 =  -\tilde \eta^t \cdot c_w(\vp^t) + \nabla \phi_w(\vp_w^t) + \eva_w \cdot \vone_w, 
        \label{eq:p-1}
    \end{split}
    \end{equation}
    where the first and second equalities come from Lemmas \ref{eg:entropy-dual} and \ref{lm:fenchel}, respectively, while the last equality again follows from Lemma \ref{eg:entropy-dual}. Plugging the above results into Equation \eqref{eq:i-1}, we obtain
    \begin{equation}
    \begin{split}
        \langle \nabla \phi_w(\vp_w^t) -  \nabla \phi_w(\tilde \vp_w),  \vp_w - \tilde \vp_w \rangle &\leq \langle \tilde \eta^t \cdot c_w(\vp^t) - \eva_w \cdot \vone_w,  \vp_w -\tilde \vp_w\rangle = \tilde \eta^t \cdot  \langle c_w(\vp^t),  \vp_w - \tilde \vp_w\rangle,  \quad \forall \vp_w \in \sQ_w^t,
    \end{split}
    \end{equation}
    where the equality holds because $\langle \eva_w \cdot \vone_w,  \vp_w - \tilde \vp_w\rangle = \eva_w \cdot (\langle \vone_w, \vp_w \rangle - \langle \vone_w, \tilde \vp_w \rangle)  = 0$.
\end{proof}

\subsection{Proof of Lemma \ref{lm:mid}}
\label{app:mid}

\begin{proof}[Proof of Lemma \ref{lm:mid}]
    For all $\vp^* \in \sP^*$, Lemma \ref{lm:triangle}  implies that 
    \begin{equation}
        D_w(\vp_w^*, \vp_w^{t + 1}) = D_w(\vp_w^*, \vp_w^t) - D_w(\vp_w^{t + 1},  \vp_w^t) + \langle \nabla \phi_w(\vp_w^{t}) - \nabla \phi_w(\vp_w^{t + 1}), \vp_w^* - \vp_w^{t + 1}\rangle.
    \label{eq:p-two-lemmas}
    \end{equation}
    Recalling that $\vs^0 < \infty$ implies $\sQ_w^t = \sP_w$, we can set $\vp_w = \vp_w^*$ in Lemma \ref{lm:iteration}, which gives rise to
    \begin{equation}
        \langle \nabla \phi_w(\vp_w^{t}) -  \nabla \phi_w(\vp_w^{t + 1}),  \vp_w - \vp_w^{t + 1} \rangle \leq \tilde \eta^t \cdot  \langle c_w(\vp^t),  \vp_w^* - \vp_w^{t + 1}\rangle.
        \label{eq:pp-3}
    \end{equation}
    Summarizing \eqref{eq:p-two-lemmas} and \eqref{eq:pp-3} for all $w \in \sW$ and combining the result leads to
    \begin{equation}
        D(\vp^*, \vp^{t + 1}) \leq D(\vp^*, \vp^t) - D(\vp^{t + 1},  \vp^t) + \tilde \eta^t \cdot  \langle c(\vp^t),  \vp^* - \vp^{t + 1}\rangle.
        \label{eq:p-m-1}
    \end{equation}
    By applying Lemmas \ref{lm:bregman-ieq} and \ref{lm:l1-l2}, we obtain the following
    \begin{equation}
        D(\vp^{t + 1}, \vp^t) \geq \frac{1}{2} \cdot \sum_{w \in \sW} \|\vp_w^t - \vp_w^{t + 1}\|_1^2 \geq \frac{1}{2} \cdot \sum_{w \in \sW} \|\vp_w^t - \vp_w^{t + 1}\|_2^2 =  \frac{1}{2} \cdot \|\vp^t - \vp^{t + 1} \|_2^2.
        \label{eq:s-1}
    \end{equation}
    Plugging \eqref{eq:s-1} into 
    \eqref{eq:p-m-1} completes the proof.
\end{proof}

\subsection{Proof of Proposition \ref{prop:decreasing}}
\label{app:decreasing}

\begin{proof}[Proof of Proposition \ref{prop:decreasing}]
    For any $\vp^* \in \sP^*$, using the results from Proposition \ref{prop:ue-vi} and Lemma \ref{lm:mid}, we can obtain
    \begin{equation}
        D(\vp^*, \vp^{t + 1}) \leq D(\vp^*, \vp^t) - \frac{1}{2} \cdot \|\vp^t - \vp^{t + 1} \|_2^2 + \tilde \eta^t \cdot  \langle c(\vp^t) - c(\vp^*),  \vp^* - \vp^{t + 1}\rangle.
        \label{eq:term-0}
    \end{equation}
    Per $1/4L$-cocoercivity of $c(\vp)$ (guaranteed by Proposition \ref{prop:property}),
    \begin{equation}
    \begin{split}
        &\langle c(\vp^t) - c(\vp^*), \vp^{t + 1} - \vp^*\rangle =  \langle c(\vp^t) - c(\vp^*), \vp^{t + 1} - \vp^t + \vp^t - \vp^*\rangle \\
        &\qquad \geq \langle c(\vp^t) - c(\vp^*), \vp^{t + 1} - \vp^t \rangle + \frac{1}{4L} \cdot \|c(\vp^t) - c(\vp^*)\|_2^2 \\
        &\qquad=-\frac{1}{L} \cdot \left( \langle c(\vp^t) - c(\vp^*), L \cdot(\vp^{t} - \vp^{t + 1}) \rangle - \frac{1}{4} \cdot \|c(\vp^t) - c(\vp^*)\|_2^2 \right) \ge -L \cdot \|\vp^t - \vp^{t + 1} \|_2^2,
        \label{eq:term-2}
    \end{split}
    \end{equation}
    where the first and second inequalities follow, respectively, from the cococercivity of $c(\vp)$ and Lemma \ref{lm:bound}. Combining Equations \eqref{eq:term-0} and \eqref{eq:term-2} then concludes the proof.
\end{proof}

\subsection{Proof of Theorem \ref{thm:convergence-ue}}
\label{app:convergence-ue}

\begin{proof}[Proof of Theorem \ref{thm:convergence-ue}]

    Proposition \ref{prop:decreasing} indicates that for any $\vp^* \in \sP^*$, the KL divergence $D(\vp^*, \vp^t)$ decreases with $t$ as long as $\tilde \eta^t = r \cdot \eta^t < 1/2L$, which is satisfied by Condition (ii) and by Conditions (i) for sufficiently large $t$. Thus, without loss of generality, we assume $\tilde \eta^t < 1/2L$ for all $t \geq 0$. It follows
    \begin{equation}
        0 \leq \frac{1 - 2\tilde \eta^t L}{2} \cdot \|\vp^t - \vp^{t + 1} \|_2^2 \leq D(\vp^*, \vp^t) - D(\vp^*,  \vp^{t + 1}).
        \label{eq:main-proof}
    \end{equation}
    Thus,  $D(\vp^*, \vp^t)$ is a monotonically decreasing sequence. According to the monotone convergence theorem, the limit of $D(\vp^*, \vp^t)$ exists for every $\vp^* \in \sP^*$. Furthermore, by letting $t \to \infty$ in \eqref{eq:main-proof} and applying the squeeze theorem, we have $\| \vp^t - \vp^{t + 1}\|_2 \to 0$. 
    
    So far, we have established that (a) $D(\vp^*, \vp^t)$ converges (\textit{not necessarily to 0}) when $t \to \infty$ for any $\vp^* \in \sP$, and (b) $\|\vp^t - \vp^{t + 1}\|$ converges to 0 when $t \to \infty$.  Both properties are necessary for the convergence of $\vp^t$ to $\sP^*$. However, taken together, they are still insufficient.  Other possibilities include but are not limited to: (a) $\vp^t$ converges before reaching a fixed point because $\eta^t$ decreases too fast; (b) $\vp^t$ converges to a fixed point that is not an equilibrium; (c) $\vp^t$ never converges, just keeping a constant ``distance" (judged by the KL divergence) with every $\vp^* \in \sP^*$.  
    
    To establish sufficiency we need to find a convergent subsequence $\{\vp^{t_j}\} \subseteq \{\vp^t\}$ with $\vp^{t_j} \to \sP^*$ when $j \to \infty$. Denote $\hat \vp \in \sP^*$ as the limit of $\vp^{t_j}$. Then, $\vp^{t_j} \to \hat \vp$ implies $D(\hat \vp, \vp^{t_j}) \to 0$. Recalling that the limit of $D(\vp^*, \vp^t)$ exists for all $\vp^* \in \sP^*$, it follows $D(\hat \vp, \vp^t) \to 0$ and hence $\vp^t \to \hat \vp$.  We proceed to prove that such a subsequence can be found under either of the two conditions stated in Theorem \ref{thm:convergence-ue}.
    
    \textbf{Condition (i).} 
    If there does \textit{not} exist a subsequence of $\{\vp^t\}$ that converges to $\sP^*$, then there must exist $\epsilon > 0$ such that $\|\vp^t - \vp^*\|_2 > \epsilon$ for all $\vp^* \in \sP^*$ and all sufficiently large $t$. Without loss of generality, we simply assume it holds for all $t \geq 0$.
    From Property (ii) in Proposition \ref{prop:property}, we have
    $\langle c(\vp^t), \vp^t - \vp^* \rangle \geq 0$, where the equality holds if and only if $\vp^t \in \sP^*$. Because $\sP^*$ is a compact set (Property (iii) in Proposition \ref{prop:property}) and the function $\langle c(\vp), \vp - \vp^*\rangle$ is continuous in $\vp$, there must exist $\delta > 0$ such that
    \begin{equation}
          \langle c(\vp^t), \vp^t - \vp^* \rangle > \delta, \quad \forall t \geq 0. 
          \label{eq:vs-bound}
    \end{equation}
    From the proof of Lemma \ref{lm:mid}, we have
    \begin{equation}
    \begin{split}
        D(\vp^*, \vp^{t + 1}) &\leq D(\vp^*, \vp^t) - \frac{1}{2} \cdot \|\vp^t - \vp^{t + 1} \|_2^2 + \tilde \eta^t \cdot  \langle c(\vp^t),  \vp^* - \vp^t + \vp^t -\vp^{t + 1}\rangle \\
        &\leq D(\vp^*, \vp^t) - \delta \cdot \tilde \eta^t - \frac{1}{2} \cdot \|\vp^t - \vp^{t + 1} \|_2^2 +  \tilde \eta^t \cdot  \langle c(\vp^t),  \vp^t -\vp^{t + 1}\rangle \\
        &\leq D(\vp^*, \vp^t) - \delta \cdot \tilde \eta^t + \frac{(\tilde \eta^t)^2}{2} \cdot \|c(\vp^t)\|_2^2 \leq  D(\vp^*, \vp^t) - \delta \cdot \tilde \eta^t + \frac{G^2}{2} \cdot (\tilde \eta^t)^2,
        \label{eq:vs-key}
    \end{split}
    \end{equation}
    where the second inequality is a direct result of \eqref{eq:vs-bound}, the third inequality follows from Lemma \ref{lm:bound}, and $G = \max_{\vp \in \sP} \|c(\vp)\|_2$ (the compactness of $\sP$ guarantees $G < \infty$).
    Telescoping Equation \eqref{eq:vs-key} yields
    \begin{equation}
         D(\vp^*, \vp^{t + 1}) \leq D(\vp^*, \vp^0) - \delta \cdot \sum_{i = 0}^t \tilde \eta^i + \frac{G^2}{2} \cdot \sum_{i = 0}^t (\tilde \eta^i)^2 = D(\vp^*, \vp^0) - \tilde \tau^t \cdot \left(\delta - \frac{G^2}{2} \cdot \tilde \gamma^t \right),
         \label{eq:decreasing-main}
    \end{equation}
    where $\tilde \tau^t = \sum_{i = 0}^t \tilde \eta^i$ and $\tilde \gamma^t = \sum_{i = 0}^t (\tilde \eta^i)^2 / \sum_{i = 0}^t \tilde \eta^i$. Per the second part of Condition (i), $\tilde \tau^t \to \infty$.  Condition (i) also ensures $\tilde \gamma^t \to 0$ (see Remark \ref{rm:rate}). Letting $t \to \infty$ in Equation \eqref{eq:decreasing-main} then leads to 
    \begin{equation}
        0 \leq \lim_{t \to \infty} D(\vp^*, \vp^{t + 1}) \leq -\infty,
    \end{equation}
    a contradiction. 
    The above proof does not apply to Condition (ii). To see this, note that if $\tilde \eta^t$ is fixed as a constant $\tilde \eta$, then we have $\tilde \gamma^t = \tilde \eta$ for all $t \geq 0$. However, it is not necessarily true that  $G^2 \cdot \tilde \eta / 2 < \delta$. Consequently, letting $t \to \infty$ in Equation \eqref{eq:decreasing-main} may not produce a contradiction. Alternatively, we establish the convergence under Condition (ii) as follows.

    \textbf{Condition (ii).}
    Since $\sP$ is a compact set, the Bolzano-Weierstrass theorem guarantees the sequence $\{\vp^t\}$ have a convergent subsequence $\{\vp^{t_j}\}$. Denoting $\hat{\vp}$ as the limit of $\vp^{t_j}$ when $j \to \infty$, we first prove for any $\hat \vs \in \bar \sR^{|\sK|}$ with $\hat \vp = q_r(
    \hat \vs)$, $\hat \vp = q_r(\hat \vs + \eta \cdot c(\hat \vp))$ (that is, $\hat \vp$ is a fixed point). Assuming there exists some $\delta > 0$ such that $\|\hat \vp- q_r(\hat \vs + \eta \cdot c(\hat \vp))\|_2 \geq \delta$. We proceed to establish a contradiction. Using the triangular inequality, we have
    \begin{equation}
    \begin{split}
        \| \vp^{t_j + 1} -  \vp^{t_j}\|_2 &=\|q_r(\vs^{t_j} + \eta \cdot c(\vp^{t_j})) -  \vp^{t_j}\|_2 \\
        &= \|q_r(\vs^{t_j} + \eta \cdot c(\vp^{t_j})) - q_r(\hat \vs + \eta \cdot c(\hat \vp)) + q_r(\hat \vs + \eta \cdot c(\hat \vp)) - \hat \vp +  \hat \vp -  \vp^{t_j}\|_2  \\
        &\geq \|q_r(\hat \vs + \eta \cdot c(\hat \vp)) - q_r(\hat \vs) \|_2 - \|
        q_r(\vs^{t_j} + \eta \cdot c(\vp^{t_j})) - q_r(\hat \vs + \eta \cdot c(\hat \vp)) +  \hat \vp -  \vp^{t_j}
        \|_2 \\
        &\geq \delta - \|  q_r(\vs^{t_j} + \eta \cdot c(\vp^{t_j})) - q_r(\hat \vs + \eta \cdot c(\hat \vp))\|_2 - \|\hat \vp -  \vp^{t_j}\|_2 \\
        &= \delta -  \|
        q_r(-\log(\vp^{t_j}) /  r + c(\vp^{t_j})) - q_r(-\log(\hat \vp) /  r + c(\hat \vp))
        \|_2 - \|\hat \vp -  \vp^{t_j}\|_2,
    \end{split}
    \label{eq:bound-h}
    \end{equation}
    where the last equality holds by Lemma \ref{lm:difference-logit}.
    Letting $j \to \infty$ on both sides of Equation \eqref{eq:bound-h}, we can first derive $\| \vp^{t_j + 1} -  \vp^{t_j}\|_2 \to 0$ on the left side. Meanwhile, on the right side, as the function $q_r(-\log(\vp) /  r + c(\vp))$ is continuous in $\vp$, $\|\hat \vp -  \vp^{t_j}\|_2 \to 0$ implies $\|q_r(-\log(\vp^{t_j}) /  r + c(\vp^{t_j})) - q_r(-\log(\hat \vp) /  r + c(\hat \vp)) \|_2 \to 0$. Thus, we arrive at $0 \geq \delta - 0 - 0 > 0$, a contradiction.
    
    According to Lemma \ref{lm:iteration}, 
    $\hat \vp = q_r(\hat \vs + \eta \cdot  c(\hat \vp))$ if and only if for all $w \in \sW$,
    \begin{equation}
    \begin{split}
        0 = \langle \nabla \phi_w(\hat \vp_w) -  \nabla \phi_w(\hat \vp_w),  \vp_w - \vp_w^{t + 1} \rangle \leq r \cdot \langle \eta \cdot c_w(\hat \vp),  \vp_w - \hat \vp_w\rangle,  \quad \forall \vp_w \in \sQ_w(\hat \vp_w),
        \label{eq:f-1}
    \end{split}
    \end{equation}
    where $\hat \sQ_w = \{\vp_w \in \sP_w: \supp(\vp_w) \subseteq \supp(\hat \vp_w)\}$.
    Summarizing Equation \eqref{eq:f-1} for all $w \in \sW$ then gives
    \begin{equation}
        \langle c(\hat \vp),  \vp - \hat \vp\rangle \geq 0, \quad \forall \vp \in \hat \sQ := \prod_{w \in \sW} \hat \sQ_w.
        \label{eq:fixed-vi}
    \end{equation}
    We then show $\hat \vs \in \sP^*$, which will conclude the proof.  We first claim $\vp^* \in \hat \sQ$; otherwise, we would have $D_{\phi}(\vp^*, \hat \vp) = \infty$, which is impossible given  $D_{\phi}(\vp^*, \vp^t)$ is  monotonically decreasing. We  then set $\vp = \vp^*$ in Equation \eqref{eq:fixed-vi}, which leads to $\langle c(\hat \vp),  \vp^* - \hat \vp\rangle \geq 0$. Meanwhile, we also have $\langle c(\hat \vp),  \hat \vp - \vp^* \rangle \geq 0$ by Property (ii) in Proposition \ref{prop:property}. Combining both gives $\langle c(\hat \vp),  \hat \vp - \vp^*\rangle = 0$, which ensures $\hat \vp \in \sP^*$ per Property (ii) in Proposition \ref{prop:property}.

    It is also worth noting that \eqref{eq:bound-h} holds only for a fixed $\eta$. Hence, the above proof is not applicable under Condition (i), which allows $\eta$ to decrease with $t$. 

    \smallskip
    \begin{remark}
    \label{rm:rate}
        If $\tilde \gamma^t$ does not converge to 0, then there must exist a subsequence $\{\tilde \gamma^{t_j}\}_{j = 0}^{\infty}$ and a constant $\xi > 0$ such that $\tilde \gamma^{t_j} \geq \xi$ for all $j \geq 0$, which  implies that
        \begin{equation}
            \sum_{i = 0}^{t_j}  \tilde \eta^i \cdot (\tilde \eta^i - \xi) = \sum_{i = 0}^{t_j} (\tilde \eta^i)^2 - \xi \cdot \sum_{i = 0}^{t_j} \tilde \eta^i \geq 0, \quad \forall j \geq 0.
            \label{eq:step-size-1}
        \end{equation}
        Meanwhile, the first part of Condition (i) guarantees that $\tilde \eta^t \to 0$. Hence, there must exist $\bar t \geq 0$ such that $\tilde \eta^t < \xi / 2$ for all $t \geq \bar t$. Without loss of generality, we assume $t_j \geq \bar t$ for all $j \geq 0$; otherwise, we can drop the first few elements of the subsequence. Then we have
        \begin{equation}
            \sum_{i = 0}^{t_j}  \tilde \eta^i \cdot (\tilde \eta^i - \xi)  \leq \sum_{i = 0}^{t_j}  \tilde \eta^i \cdot (\xi / 2 - \xi) = -\xi / 2 \cdot \sum_{i = 0}^{t_j}  \tilde \eta^i, \quad \forall j \geq 0.
            \label{eq:step-size-2}
        \end{equation}
        Letting $j \to \infty$ in Equation \eqref{eq:step-size-2} lead to $ \sum_{i = 0}^{t_j}  \tilde \eta^i \cdot (\tilde \eta^i - \xi) \to -\infty$, which contradicts Equation \eqref{eq:step-size-1}.
    \end{remark}
\end{proof}

\end{appendices}

\section*{Acknowledgements}

This research is funded by US National Science Foundation under the awards  CMMI \#2225087 and ECCS \#2048075.  The authors are grateful for the valuable comments offered by Mr. Boyi Liu and Prof. Hani Mahmassani at Northwestern University and Prof. Yafeng Yin at the University of Michigan, Ann Arbor. The remaining errors are our own.

\bibliographystyle{ims}
\begin{small}
\bibliography{example_paper}
\end{small}

\end{document}